\documentclass[12pt,letterpaper,american]{article}

\usepackage[utf8]{inputenc}
\usepackage[T1]{fontenc} 
\usepackage{babel} 
\usepackage{lmodern} 
\usepackage{microtype}

\usepackage{comment}

\usepackage{titlesec}
\titlespacing*{\section}{0pt}{2.0ex plus 0.6ex minus 0.4ex}{1.0ex plus 0.3ex}
\titlespacing*{\subsection}{0pt}{1.5ex plus 0.5ex minus 0.3ex}{0.8ex plus 0.2ex}
\titlespacing*{\subsubsection}{0pt}{1.2ex plus 0.4ex minus 0.2ex}{0.6ex plus 0.2ex}
\titlespacing*{\paragraph}{0pt}{1.0ex plus 0.3ex minus 0.2ex}{0.8em}

\usepackage{titling}

\pretitle{\begin{center}\bfseries\Large}
\posttitle{\par\end{center}\vspace{-1.5em}}

\preauthor{\begin{center}\large\lineskip .5em\begin{tabular}[t]{c}}
\postauthor{\end{tabular}\par\end{center}\vspace{-1.5em}}

\predate{\begin{center}\small}
\postdate{\par\end{center}\vspace{-1.0em}}

\usepackage[letterpaper,includeheadfoot]{geometry}

\usepackage{setspace} 
\usepackage[authoryear]{natbib} 
\usepackage{amsmath,amssymb,amsthm,thmtools,mathtools,mathrsfs} 
\usepackage{enumitem,algorithm2e} 

\setlist{
  topsep=2pt,
  itemsep=1pt,
  parsep=0pt,
  partopsep=0pt
}

\makeatletter
\def\thm@space@setup{%
  \thm@preskip=6pt plus 2pt minus 2pt
  \thm@postskip=6pt plus 2pt minus 2pt
}
\makeatother

\usepackage{tabularx, booktabs} 

\usepackage[normalem]{ulem} 
\usepackage{xcolor} 

\usepackage{graphicx}
\usepackage{subcaption}
\usepackage{wasysym}
\usepackage{tikz}

\usepackage{listings}
\usepackage[left, displaymath, mathlines, pagewise]{lineno}

\newlength{\lyxlabelwidth}      

\providecommand{\factname}{Fact}

\theoremstyle{remark}
\newtheorem{claim}{Claim}
\theoremstyle{plain}
\newtheorem{proposition}{Proposition}
\theoremstyle{plain}
\newtheorem{fact}{\protect\factname}
\newtheorem{lemma}{Lemma}
\newtheorem{theorem}{Theorem}
\newtheorem{corollary}{Corollary}
\theoremstyle{definition}
\newtheorem{example}{Example}

\newtheorem{remark}{Remark}
\newtheorem{definition}{Definition}

\newtheorem*{proposition*}{Proposition}

\usepackage{centernot}

\newtheoremstyle{named} 
{} 
{} 
{} 
{} 
{\bfseries} 
{:} 
{.5em} 
{#3} 

\newcommand{\casename}{Case}
\newcommand{\subcasename}{Subcase}

\newlist{casenv}{enumerate}{4}
\setlist[casenv]{leftmargin=*,align=left,widest={iiii}}
\setlist[casenv,1]{label={{\textit\ \casename} \arabic*.},ref=\arabic*}
\setlist[casenv,2]{label={{\textit\ \subcasename} \roman*.},ref=\roman*}
\setlist[casenv,3]{label={{\textit\ \subcasename\ \alph*.}},ref=\alph*}
\setlist[casenv,4]{label={{\textit\ \subcasename} \arabic*.},ref=\arabic*}

\newif\ifshowcomments
\showcommentstrue    

\ifshowcomments
  \newcommand{\DF}[1]{{\color{blue}[DF: #1] }}
  \newcommand{\JC}[1]{{\color{orange}[JC: #1]}}
\else
  \newcommand{\DF}[1]{}
  \newcommand{\JC}[1]{}
\fi

\newcommand{\del}[1]{\textcolor{blue}{\xout{#1}}}
\newcommand{\add}[1]{\textcolor{red}{\uwave{#1}}}

\newcommand{\defterm}[1]{\emph{#1}}
\newcommand{\defprop}[1]{\emph{#1}}

\usepackage{xspace}
\newcommand{\prop}[1]{\textsc{#1}}

\DeclareRobustCommand{\BAL}{\prop{BAL}\xspace}
\DeclareRobustCommand{\IR}{\prop{IR}\xspace}
\DeclareRobustCommand{\WELB}{\prop{WELB}\xspace}
\DeclareRobustCommand{\SP}{\prop{SP}\xspace}
\DeclareRobustCommand{\TP}{\prop{TP}\xspace}   
\DeclareRobustCommand{\DSP}{\prop{DSP}\xspace} 
\DeclareRobustCommand{\SDSP}{\prop{SDSP}\xspace}
\DeclareRobustCommand{\MAR}{\prop{MAR}\xspace}
\DeclareRobustCommand{\PE}{\prop{PE}\xspace}
\DeclareRobustCommand{\IGE}{\prop{IGE}\xspace}

\DeclareRobustCommand{\acksize}{\fontsize{9}{11}\selectfont}

\usepackage[unicode=true,pdfusetitle,
 bookmarks=true,bookmarksnumbered=false,bookmarksopen=false,
 breaklinks=true,pdfborder={0 0 0},pdfborderstyle={},backref=false,colorlinks=true]
 {hyperref}
\hypersetup{linkcolor=blue, citecolor=blue, urlcolor=blue}

\title{Axioms for Top Trading Cycles in Multi-Object Reallocation\thanks{\acksize 
We are especially grateful to Ivan Balbuzanov for his guidance and support throughout this project. We also thank Açelya Altuntaş, Onur Kesten, Bettina Klaus, Alexandru Nichifor, Szilvia Pápai, Tom Wilkening, and Steven Williams, as well as seminar participants at HEC Lausanne, Deakin University, the 2024 Economic Theory Festival at UNSW, ESAM 2024, and MATCH-UP 2024 for their helpful comments and suggestions. Feng gratefully acknowledges financial support from the National Natural Science Foundation of China No. 72503028. Coreno received financial support from an Australian Government RTP Scholarship, a University of Melbourne FBE-GRATS Scholarship, and an M. A. Bartlett Scholarship. Coreno and Feng gratefully acknowledge financial support from the Swiss National Science Foundation (SNSF) through Project 100018-231777. This article is based upon work supported by the National Science Foundation under Grant No.~DMS-1928930 and by the Alfred P.~Sloan Foundation under grant G-2021-16778, while authors were in residence at the Simons Laufer Mathematical Sciences Institute (formerly MSRI), Berkeley, CA, Fall 2023.}}

\author{Jacob Coreno\textsuperscript{a,}\thanks{Correspondence to: Faculty of Business and Economics, University of Lausanne, Internef 530, 1015 Lausanne,
Switzerland.}
\and 
Di Feng\textsuperscript{b}}

\begin{document}
\maketitle

\renewcommand{\thefootnote}{\alph{footnote}}

\footnotetext[1]{University of Lausanne, Switzerland. Contact: \protect\href{mailto:jacob.coreno@unil.ch}{jacob.coreno@unil.ch}.}
\footnotetext[2]{Dongbei University of Finance and Economics, China. Contact: \protect\href{mailto:dfeng@dufe.edu.cn}{dfeng@dufe.edu.cn}.}

\vspace{-2.0em}
\begin{abstract}
This paper studies multi-object reallocation without monetary transfers, where agents initially own multiple indivisible objects and have strict preferences over bundles (e.g., shift exchange among workers at a firm). Focusing on marginal rules that elicit only rankings over individual objects, we provide axiomatic characterizations of the generalized Top Trading Cycles rule (TTC) on the lexicographic and responsive domains. On the lexicographic domain, TTC is characterized by balancedness, individual-good efficiency, the worst-endowment lower bound, and either truncation-proofness or drop strategy-proofness. On the responsive domain, TTC is the unique marginal rule satisfying individual-good efficiency, truncation-proofness, and either the worst-endowment lower bound or individual rationality. In the Shapley--Scarf housing market, TTC is characterized by Pareto efficiency, individual rationality, and truncation-proofness. Finally, on the conditionally lexicographic domain, the augmented Top Trading Cycles rule is characterized by balancedness, Pareto efficiency, the worst-endowment lower bound, and drop strategy-proofness. The conditionally lexicographic domain is a maximal domain on which Pareto efficiency coincides with individual-good efficiency.
\end{abstract}

\noindent \textbf{Keywords:} Top Trading Cycles; multi-object exchange; responsive preferences; lexicographic preferences; truncation-proofness; worst-endowment lower bound.

\noindent \textbf{JEL Classification:} C78; D47; D71.

\section{Introduction}\label{sec:introduction}
\renewcommand{\thefootnote}{\arabic{footnote}}

We study \defterm{multi-object reallocation} without monetary transfers.
A problem instance consists of a group of agents, an initial endowment in which each agent owns a set of heterogeneous and indivisible objects, and each agent's strict preferences over bundles. A \defterm{rule} maps reported preferences to a feasible reassignment. When each agent is endowed with a single object, the model reduces to the classic \defterm{housing market} of \citet{shapley1974}.

Reallocation problems arise in a variety of applications. 
Firms plan shift schedules months in advance, and employees exchange their assigned shifts (e.g., \citet{manjunath2021}). Universities exchange students through programs such as the Tuition Exchange Program in the United States and Erasmus in Europe
(e.g., \citet{dur2019two,bloch2020}). Living-donor organ exchange programs provide further examples featuring both single- and multi-object exchange (e.g., \citet{roth2004,roth2005,ergin2017dual,ergin2020efficient}).\footnote{Related work studies course allocation as a \emph{multi-object assignment} problem without endowments (e.g., \citet{budish2011combinatorial,budish2012,bichler2021}).}

Multi-object reallocation confronts two difficulties that are largely absent in single-object reallocation.
The first is the vast number of possible bundles, which makes it infeasible for agents to accurately report their preferences over bundles in practice.\footnote{For instance, in shift reallocation with only 20 shifts, there are ${20 \choose 5} = 15{,}504$ distinct bundles of five shifts.} As \citet[p.~331]{roth2015experiments} explains, ``a practical mechanism must simplify the language in which preferences can be reported, and by doing so it will restrict which preferences can be reported.'' Accordingly, we focus primarily on \defprop{marginal} rules, under which agents report only rankings over individual objects (their \defterm{marginal preferences}).\footnote{This reporting format is used, for example, in the National Resident Matching Program (NRMP): hospitals submit ``rank-order lists'' over individual doctors despite having potentially complex preferences over sets of doctors (e.g., \citet{roth1999redesign}, \citet{milgrom2009assignment,milgrom2011critical}).}

When working with marginal rules, we focus on two preference domains in which marginal preferences are informative about the underlying bundle comparisons. On the lexicographic domain, marginal rules are without loss of generality because an agent's ranking over individual objects fully determines her preferences over bundles. On the more general responsive domain (\citet{roth1985college}), marginal preferences remain informative---especially when comparing bundles that differ by one object---but they need not determine all possible bundle comparisons. Both domains rule out complementarities, which cannot be expressed via marginal preferences.\footnote{To address complementarities, Section~\ref{Results for CL preferences} considers a richer but still succinct reporting language based on ``lexicographic preference trees''.}

The second difficulty is the tension among standard objectives relating to efficiency, voluntary participation, and incentive compatibility.
While the ideal desiderata---\defprop{Pareto efficiency (\PE)}, \defprop{individual rationality (\IR)}, and \defprop{strategy-proofness (\SP)}---are jointly attainable in housing markets,\footnote{In housing markets, these properties uniquely characterize the Top Trading Cycles rule (\citet{ma1994}).} they become incompatible once some agent owns more than one object (\citet{sonmez1999}).
Much of the axiomatic literature sidesteps this incompatibility by insisting on strategy-proofness and sacrificing one of the other desiderata. Strategy-proofness is especially restrictive in multi-object environments.
Under mild auxiliary axioms, combining \SP with \PE leads to \defterm{Sequential Dictatorships} (\citet{papai2001,klaus2002,ehlers2003scw,hatfield2009strategy}), while combining \SP with \IR leads to \defterm{Segmented Trading Cycles rules} (\citet{papai2003,papai2007}). Either way, insisting on strategy-proofness forces major concessions---either by effectively ignoring the initial ownership structure
or by accepting severe inefficiency.

In this paper, we circumvent the incompatibility by making modest, practically motivated compromises relative to the ideal notions of efficiency and incentive compatibility. We study the \defterm{generalized Top Trading Cycles rule (TTC)}, a marginal rule based on a multi-object extension of the celebrated Top Trading Cycles algorithm (attributed to David Gale by \citet{shapley1974}).
We show that TTC comes close to the unattainable ideals: it delivers meaningful guarantees regarding efficiency, participation, and incentive compatibility. More concretely, we characterize TTC on various domains using \defprop{balancedness (\BAL)}, \defprop{individual-good efficiency (\IGE)}, the participation guarantees \IR and the \defprop{worst-endowment lower bound (\WELB)}, and the incentive requirements \defprop{truncation-proofness (\TP)} and \defprop{drop strategy-proofness (\DSP)}.

On the lexicographic domain, TTC is characterized by \BAL, \PE, \WELB, and either \TP or \DSP (Theorems~\ref{Multi-unit characterization with TP} and~\ref{Multi-unit characterization with drop SP}).
On the responsive domain, where no marginal rule can satisfy \IGE, \IR, and \DSP (Corollary~\ref{impossibility1}), TTC is characterized as the unique marginal rule satisfying \IGE, \TP, and either \WELB or \IR (Theorems~\ref{responsive characterization} and~\ref{responsive characterization with IR}).
In housing markets, our axioms simplify, and TTC is characterized by \PE, \IR, and \TP (Theorem~\ref{Theorem: Single-unit characterization}). We briefly discuss these properties below, deferring formal definitions and further intuition to Section~\ref{sec:Properties}. 

Balancedness (\BAL) requires that
each agent receive as many objects as in her endowment. In many applications this is imposed as a feasibility constraint.\footnote{In shift reallocation it may reflect contractual constraints on the number of shifts per worker (\citet{manjunath2021,manjunath2024}). Similar constraints arise in time banks and exchange programs (\citet{andersson2021organizing,biro2022,dur2019two}).} We treat \BAL as an axiom, rather than a model restriction, so that we can study its interaction with other properties (see Lemmata~\ref{individual-good + IR implies WELB} and~\ref{lemma:BAL implied on responsive domain}).

Individual rationality (\IR) and the worst-endowment lower bound (\WELB) are distinct participation guarantees. \IR ensures that the outcome weakly Pareto dominates the initial allocation; that is, every agent enjoys a bundle at least as good as her endowment. \WELB is an object-wise guarantee that ensures no agent receives any object ranked below the worst in her endowment according to her marginal preferences. \IR and \WELB coincide in single-object reallocation. In multi-object problems the properties diverge, but \IR implies \WELB under mild additional structure.\footnote{See Lemma~\ref{individual-good + IR implies WELB} and the alternative preference domains in the surrounding discussion.} TTC satisfies both guarantees on each domain we study.

For efficiency, we use Pareto efficiency (\PE) when it is attainable and individual-good efficiency (\IGE) as a tractable relaxation. \IGE relaxes \PE by ruling out mutually beneficial exchanges where each participating agent trades a single object. Such trades are comparatively easy to identify and execute, and under responsive preferences they can be detected using only marginal preferences. The two properties coincide in single-object reallocation and multi-object reallocation with lexicographic preferences. On the responsive domain, the relaxation is substantive: full \PE is computationally demanding (e.g., \citet{de2009complexity,aziz2019efficient}) and, among marginal rules, incompatible with \IR (\citet{manjunath2024}).

Drop strategy-proofness (\DSP) and truncation-proofness (\TP) protect against simple manipulations of agents' marginal preferences. 
An agent implements a \defterm{drop strategy} by demoting an object she does not own to the bottom of her marginal preference list (\citet{altuntacs2023}); she implements a \defterm{truncation strategy} by demoting every such object that she ranks below a chosen cutoff. Intuitively, truncation strategies can be thought of as ``shortening the list of acceptable objects'' in the standard sense from two-sided matching (e.g., \citet{mongell1991sorority,roth1999truncation,ehlers2008}). \DSP and \TP require that no agent can improve her assignment using the corresponding class of manipulations.

In addition to marginal preferences, we also consider a richer---yet still succinct---reporting language that can express complementarities. We study \defterm{conditionally lexicographic} preferences, which admit a compact representation via ``lexicographic preference trees'' and preserve the equivalence between \IGE and \PE found under lexicographic preferences. On this domain, we study the \defterm{Augmented Top Trading Cycles rule (ATTC)} introduced by \citet{fujita2018}. We show that ATTC satisfies natural extensions of our key properties and is characterized by \BAL, \PE, \WELB, and \DSP (Theorem~\ref{characterization of ATTC}).
Finally, we show that conditionally lexicographic preferences form a maximal domain on which \PE and \IGE are equivalent (Proposition~\ref{maximal domain 1}); beyond this boundary, ruling out Pareto-improving single-object exchanges no longer guarantees full Pareto efficiency.

The paper proceeds as follows. Section~\ref{sec:model} introduces the model, preference domains, properties, and TTC. Sections~\ref{results for lexicographic preferences} and~\ref{sec:responsivedomain} present our characterizations on the lexicographic and responsive domains, and Section~\ref{sec:SS} specializes the results to housing markets. Section~\ref{Results for CL preferences} studies conditionally lexicographic preferences and characterizes ATTC. Section~\ref{related literature} discusses related literature. The appendices formalize lexicographic preference trees, contain omitted proofs, and provide additional examples.

\section{Model} 
\label{sec:model}
Let $N=\left\{ 1,\dots,n\right\} $ be a finite set of $n \geq 2$ \defterm{agents}.
Let $O$ be a finite set of heterogeneous and indivisible \defterm{objects} such that $|O| \geq n$. A \defterm{bundle} is a subset of $O$. Let $2^{O}$ denote the set of bundles. We denote generic elements of $O$ by lowercase letters (e.g., $x,y,z$), and generic elements of $2^O$ by uppercase letters (e.g., $X,Y,Z$). To simplify notation, when there is no risk of confusion, we identify a
singleton set $\left\{ x\right\}$ with the element $x$ itself. For example, we write $X \cup x$ to denote $X \cup \{x\}$.

An \defterm{allocation} is a function $\mu:N\to2^{O}$ such that (i) for all $i \in N$, $\mu(i) \neq \emptyset$, (ii)
for all $i,j\in N$, $i\neq j$ implies $\mu\left(i\right)\cap\mu\left(j\right)=\emptyset$,
and (iii) $\bigcup_{i\in N}\mu\left(i\right)=O$. Thus, an allocation $\mu$ can be represented as a profile $\left(\mu_{i}\right)_{i\in N}$
of nonempty, pairwise disjoint bundles satisfying $\bigcup_{i\in N}\mu_{i}=O$. For each $i\in N$,
$\mu_{i}$ is called agent $i$'s \defterm{assignment} at $\mu$. Let $\mathcal{A}$ denote the set of allocations.

The \defterm{initial allocation}, also referred to as the \defterm{endowment allocation}, is denoted by $\omega=\left(\omega_{i}\right)_{i\in N}$.
For each $i\in N$, $\omega_{i}$ is called agent~$i$'s
\defterm{endowment}. The initial \defterm{owner} of an object $o$ is the unique agent $i$ with $o \in \omega_i$.

Each agent $i$ has a \defterm{(strict) preference} relation $P_{i}$ on the set of bundles. We assume that $P_i$ belongs to some specified subset $\mathcal{P}_i$ of all strict preference relations on $2^O$.
In subsequent sections, we impose further structure on the sets $\mathcal{P}_i$. If agent $i$ prefers bundle $X$
to $Y$, then we write $X \mathrel{P_i} Y$. Let $R_{i}$ denote the
``at least as good as'' relation associated with $P_{i}$, defined by
$X \mathrel{R_i} Y$ if and only if ($X \mathrel{P_i} Y$ or $X=Y$).\footnote{Formally, $R_{i}$ is a linear order (i.e., a \emph{complete}, \emph{transitive}, and \emph{antisymmetric}
binary relation) on $2^O$ , and $P_{i}$ is the strict (i.e., \emph{irreflexive} and \emph{asymmetric}) part of $R_{i}$.} Given a nonempty bundle
$X\in2^{O}$, $\max_{P_{i}}\left(X\right)$ denotes the most-preferred
object in $X$ at $P_i$, i.e., $\max_{P_{i}}\left(X\right)=x$ if $x\in X$
and $x \mathrel{R_i} y$ for all $y\in X$. Similarly, $\min_{P_{i}}\left(X\right)$
denotes the least-preferred object in $X$ at $P_i$, i.e., $\min_{P_{i}}\left(X\right)=x$ if
$x\in X$ and $y \mathrel{R_i} x$ for all $y\in X$. A \defterm{preference profile} is an indexed family $P=\left(P_{i}\right)_{i\in N}$ of preference relations. The \defterm{domain} is the set $\mathcal{P} \coloneqq \prod_{i \in N} \mathcal{P}_i$, representing all possible preference profiles under consideration.

An object reallocation problem (or simply a \defterm{problem}) is a triple $(N,\omega,P)$. Since $(N,\omega)$ remains fixed throughout, we identify a problem with its preference profile $P$. Thus, the domain $\mathcal{P} = \prod_{i \in N} \mathcal{P}_i$ of preference profiles represents the set of all problems.

A \defterm{rule (on $\mathcal{P}$)} is a function $\varphi:\mathcal{P} \to \mathcal{A}$ that associates with each preference profile $P$ an allocation
$\varphi\left(P\right)$. For each $i\in N$, $\varphi_{i}\left(P\right)$
denotes agent $i$'s assignment at $\varphi\left(P\right)$.

\subsubsection*{Marginal rules}
A rule is called \defprop{marginal} if it can be implemented with a simple reporting language consisting of linear orders over individual objects. Marginal rules are attractive in applications because they ask agents only for a ranking over individual objects, rather than a ranking over all feasible bundles, and thus substantially reduce the informational and cognitive burden on participants. A prominent example is the National Resident Matching Program (NRMP), in which hospitals submit rank-order lists over individual doctors, even though they may have complex preferences over sets of doctors (see, e.g., \citet{roth1999redesign}, \citet{roth2002economist}, \citet{milgrom2009assignment,milgrom2011critical}).

Formally, given a preference relation $P_i$ on $2^O$, the \defterm{marginal preference} over a subset $X \subseteq O$ of individual objects, denoted $P_i |_X$, is the restriction of $P_i$ to singleton subsets of $X$.\footnote{Equivalently, $P_i |_X$ is the strict linear order on $X$ such that, for all $x,y \in X$, $x \mathrel{P_i |_X} y$ if and only if $x \mathrel{P_i} y$.} Similarly, $R_i |_X$ denotes the restriction of $R_i$ to singleton subsets of $X$. We often represent a marginal preference $P_i |_X$ as an ordered list of objects; for example, $P_i |_X:x_1,x_2,\dots, x_{|X|}$ means that $x_1 \mathrel{P_i} x_2 \mathrel{P_i} \cdots \mathrel{P_i} x_{|X|}$, and $P_i |_X:x_1, x_2, \dots, x_k, \dots$ means that $x_1 \mathrel{P_i} x_2 \mathrel{P_i} \cdots \mathrel{P_i} x_k \mathrel{P_i} o$ for all $o \in X \setminus \{x_1 , x_2 ,\dots, x_k\}$. 
Given a preference profile $P = (P_i)_{i \in N}$,
the \defterm{marginal preference profile} over $X$ is the profile $P |_X = (P_i |_X)_{i \in N}$.

A rule is marginal if it depends solely on agents' marginal preferences over $O$.

\begin{definition}
    A rule $\varphi$ is \defprop{marginal (MAR)} if, for all $P, P' \in \mathcal{P}$, $P|_O = P'|_O$ implies $\varphi(P) = \varphi(P')$.
\end{definition}

Marginal rules are most natural in settings where objects behave like substitutes, so that agents' preferences over bundles are captured reasonably well by their rankings over individual objects. Most of our analysis focuses on these settings (Sections~\ref{results for lexicographic preferences}--\ref{sec:SS}). The main limitation of marginal rules is that they do not allow agents to express complementarities across objects. To address complementarities, Section~\ref{Results for CL preferences} (and Appendix~\ref{Appendix: LP trees}) considers a richer but still succinct reporting language based on ``lexicographic preference trees''.

\subsection{Preference domains}\label{preference domains}

We now describe the lexicographic and responsive domains, two structured preference domains where objects behave like substitutes.

An agent has lexicographic preferences if, when evaluating distinct bundles $X$ and $Y$, she prefers the bundle containing the most-preferred object in $X \cup Y$; if the most-preferred object in $X \cup Y$ is common to $X$ and $Y$, then she prefers the bundle containing the second-most-preferred object in $X \cup Y$, and so on.\footnote{Similar ``take the best'' decision heuristics have been documented in psychology (e.g., \citet{gigerenzer1996reasoning}, \citet{gigerenzer1999simple}).} Formally, a preference relation $P_i$ on $2^O$ is \defterm{lexicographic} if, for any two distinct bundles $X$ and $Y$, 
$$
X \mathrel{P_i} Y \iff \max_{P_i}\left(X \mathbin{\triangle} Y\right) \in X,$$
where $X \mathrel{\triangle} Y = (X \setminus Y) \cup (Y \setminus X)$ is the symmetric difference between $X$ and $Y$. 
For each $i \in N$, let $\mathcal{L}_i$ be the set of lexicographic preferences on $2^{O}$, and let $\mathcal{L} \coloneqq \prod_{i \in N} \mathcal{L}_i$ be the \defterm{lexicographic domain}.

Although it is somewhat restrictive, the lexicographic domain is a natural starting point in our analysis. Each lexicographic preference $P_i \in \mathcal{L}_i$ is uniquely determined by its marginal preference $P_i |_O$ over individual objects; that is, for all $P_i, P'_i \in \mathcal{L}_i$, $P_i |_O = P'_i |_O$ if and only if $P_i = P'_i$. Consequently, any rule defined on $\mathcal{L}$ is automatically marginal. We therefore identify each $P_i \in \mathcal{L}_i$ with its associated marginal preference $P_i |_O$ and write $P_{i}:x_{1},x_{2},\dots,x_{|O|}$ if $x_{1} \mathrel{P_i} x_{2} \mathrel{P_i} \cdots \mathrel{P_i} x_{|O|}$.

Responsiveness is a more general condition first studied by \citet{roth1985college} for many-to-one matching models. An agent has responsive preferences if, for any two bundles that differ in one object, she prefers the bundle containing the more-preferred object. Formally, a preference relation $P_i$ on $2^O$ is \defterm{responsive} if, for any bundle $X$ and any objects $y,z\in O \setminus X$, 
$$
(X\cup y) \mathrel{P_i} (X\cup z) \iff y\mathrel{P_i} z.
$$
Responsive preferences rule out complementarities because the relative ranking between any two objects is independent of the other objects they are obtained with. For each $i \in N$, let $\mathcal{R}_i$ denote the set of responsive preferences on $2^O$, and let $\mathcal{R} \coloneqq \prod_{i \in N} \mathcal{R}_i$ be the \defterm{responsive domain}. 

It is helpful to compare the lexicographic and responsive domains with the domain of monotonic preferences. A preference relation $P_i$ on $2^O$ is \defterm{monotonic} if, for any bundles $X$ and $Y$, $X \mathrel{R_i} Y$ whenever $X \supseteq Y$. For each $i \in N$, let $\mathcal{M}_i$ be the set of monotonic preferences on $2^O$, and let $\mathcal{M} \coloneqq \prod_{i \in N} \mathcal{M}_i$ be the \defterm{monotonic domain}.

Every lexicographic preference is both responsive and monotonic, while responsiveness and monotonicity are logically independent. Formally, for each $i \in N$, we have
$\mathcal{L}_i \subseteq \mathcal{R}_i \cap \mathcal{M}_i$ (with strict inclusion if $|O| \geq 3$), 
$\mathcal{M}_i \nsubseteq \mathcal{R}_i$ (whenever $|O| \geq 3)$, and $\mathcal{R}_i \nsubseteq \mathcal{M}_i$.\footnote{\label{independence of responsive and monotonic domains}For example, if $O = \{a,b\}$, then $P_i:\{a\}, \emptyset, \{a, b\}, \{b\}$ and $P_i': \emptyset, \{a\}, \{b\}, \{a,b\}$ are responsive but not monotonic. If $O = \{a, b, c\}$, then $P'_i: \{a, b, c\}, \{a, c\}, \{a, b\}, \{b, c\}, \{a\} , \{b\}, \{c\}, \emptyset$ is monotonic but not responsive.}

\subsection{Properties of allocations and rules}\label{sec:Properties}
This section introduces several desirable properties of allocations and rules. Unless stated otherwise, these properties are defined for an arbitrary domain $\mathcal{P}$ of strict preferences over bundles. Our main focus will be the lexicographic and responsive domains introduced above.

Our first requirement is that each agent ends up with the same number of objects as initially endowed. An allocation $\mu$ is \defprop{balanced} if, for each agent $i \in N$, $|\mu_i| = |\omega_i|$.

\begin{definition}
\label{def: balancedness}
    A rule $\varphi$ satisfies \defprop{balancedness (\BAL)} if, for each $P \in \mathcal{P}$, $\varphi(P)$ is balanced.
    \end{definition}

Balancedness is often treated as an inviolable constraint in practical reallocation problems. In shift reallocation, for instance, it prevents overwork or underemployment and may reflect contractual constraints on the number of shifts per worker (\citet{manjunath2021,manjunath2024}). In student and tuition exchange programs, \BAL helps to maintain reciprocity and control education costs at popular schools (\citet{andersson2021organizing},  \citet{biro2022}, \citet{dur2019two}).\footnote{\citet{dur2019two} document several cases in which such long-run imbalances led to the failure of an exchange program.} We model \BAL as a property of rules, rather than as an explicit feasibility constraint, so that we can study its interaction with other properties (see Lemmata~\ref{individual-good + IR implies WELB} and~\ref{lemma:BAL implied on responsive domain}).

\subsubsection{Efficiency}\label{sec:Efficiency}
An allocation $\overline{\mu}$ \defterm{Pareto dominates} another allocation $\mu$ at a preference profile $P$ if (i) for all $i\in N$, $\overline{\mu}_{i} \mathrel{R_i} \mu_{i}$, and (ii) for some $i\in N$, $\overline{\mu}_{i} \mathrel{P_i} \mu_{i}$. An allocation $\mu $ is \defprop{Pareto efficient} at $P$ if it is not Pareto dominated at $P$ by any allocation.

\begin{definition}
\label{def: Pareto efficiency}
    A rule $\varphi$ satisfies \defprop{Pareto efficiency (\PE)} if, for each $P \in \mathcal{P}$, $\varphi(P)$ is Pareto efficient at $P$.
\end{definition}

In principle, any allocation that violates Pareto efficiency could be destabilized by a coalition of agents carrying out a mutually beneficial exchange. In practice, however, such exchanges may be hard to realize. A coalition would typically need detailed information about the agents' preferences over bundles and a way to coordinate intricate trades of multiple objects among many agents. Even with the required preference information, finding a Pareto improvement---or certifying that none exists---is computationally demanding. In particular, under additive preferences (a subclass of responsive preferences), deciding whether a Pareto improvement exists is NP-complete (e.g., \citet{de2009complexity,aziz2019efficient}).

Individual-good efficiency relaxes \PE by ruling out a restricted class of mutually beneficial exchanges---multilateral trades where each participating agent offers a single object in return for another.
Formally, a \defterm{trading cycle} at an allocation $\mu$ is a cyclic sequence
\[
C=\left(o_1, i_{1},o_2,i_{2},\dots,i_{k-1},o_k,i_{k},o_{k+1}=o_1\right)
\]
consisting of $k \geq 1$ distinct objects and $k$ distinct agents such that, for each $\ell \in \{1, \dots, k\}$, $o_\ell \in \mu_{i_\ell}$. For each $\ell \in \{1, \dots, k\}$, we interpret $o_\ell$ as the object relinquished by agent~$i_\ell$ and $o_{\ell + 1}$ as the object received. We say that $C$ is \emph{Pareto improving} at a preference profile $P$ if
\[
\text{for each } \ell \in \{1, \dots, k\}, \quad (\mu_{i_\ell} \cup o_{\ell + 1}) \setminus o_\ell \mathrel{P_{i_\ell}} \mu_{i_\ell}.
\]
An allocation $\mu$ is \defterm{individual-good efficient} at $P$ if no trading cycle at $\mu$ is Pareto improving at $P$.\footnote{This terminology is borrowed from \citet{biro2022}. Similar properties are studied in \citet{aziz2019efficient}, \citet{caspari2020}, and \citet{coreno2022}.} 

\begin{definition}
   A rule $\varphi$ satisfies \defprop{individual-good efficiency (\IGE)} if, for each $P \in \mathcal{P}$, $\varphi(P)$ satisfies individual-good efficiency at $P$. 
\end{definition}

Although \PE and \IGE coincide on the lexicographic domain (\citet{aziz2019efficient}), on more general preference domains an allocation may satisfy \IGE even when further Pareto improvements are possible through more complex trades. This loss of efficiency is the price we pay for a notion that is tractable and achievable in practice. \IGE focuses on single-object exchanges, which are relatively easy to organize,\footnote{A similar relaxation is ``pair efficiency'' (\citet{ekici2022}), which precludes only mutually beneficial \emph{bilateral} trades.} and under responsive preferences such exchanges can be detected using only the information elicited by a marginal rule. In addition, there are polynomial-time algorithms for deciding whether a given allocation satisfies \IGE on each of the domains we consider (e.g., \citet{cechlarova2014pareto,aziz2019efficient,fujita2018}).

\subsubsection{Participation guarantees}
An allocation $\mu$ is \defprop{individually rational} at a preference profile $P$ if, for each $i \in N$, $\mu_i \mathrel{R_i} \omega_i$.

\begin{definition}
\label{def: individual rationality}
    A rule $\varphi$ satisfies \defprop{individual rationality (\IR)} if, for each $P \in \mathcal{P}$, $\varphi(P)$ is individually rational at $P$.
    \end{definition}

Individual rationality is a participation guarantee that ensures no agent is made worse off by taking part in the reallocation. Our second participation guarantee restricts which individual objects may appear in an agent's assignment. An allocation $\mu$ satisfies the \defprop{worst-endowment lower bound} at a preference profile $P$ if, for each $i\in N$ and each $o \in \mu_i$, $o\mathrel{R_i} \min_{P_{i}}\left(\omega_i \right)$; that is, no agent is assigned an object that she ranks below her least-preferred endowed object.

\begin{definition}
\label{def: worst-endowment lower bound}
A rule $\varphi$ satisfies the \defprop{worst-endowment lower bound (\WELB)} if, for each $P \in \mathcal{P}$, $\varphi(P)$
satisfies the worst-endowment lower bound at $P$.
\end{definition}

Note that \WELB is formulated purely in terms of agents' marginal preferences over individual objects, whereas \IR is defined in terms of preferences over bundles.

\IR and \WELB coincide in the single-object environment, where each agent is endowed with a single object. In general multi-object problems, however, the two properties are logically independent. Intuitively, \IR permits an agent to receive an individually unattractive object as part of a sufficiently desirable bundle, whereas \WELB rules out such objects altogether. One can interpret \WELB as giving agents veto power over any individual object owned by others: by ranking an object below every object in her endowment, an agent can ensure that she is never assigned that object.

Under additional structure, \IR becomes strictly stronger than \WELB. For example, on the responsive domain, any marginal and \IR rule automatically satisfies \BAL and \WELB.

\begin{lemma}\label{individual-good + IR implies WELB}
    On the responsive domain, if a marginal rule satisfies \IR, then it satisfies \BAL and \WELB.
\end{lemma}

A similar implication holds on any domain $\mathcal{P}^*$ that treats bundles that would violate an agent's \WELB condition as ``unacceptable'' (see, e.g., \citet{andersson2021organizing}, \citet{biro2022}).\footnote{\label{footnote other domain}Formally, consider a domain $\mathcal{P}^*$ such that, for each $i \in N$ and each $P_i \in \mathcal{P}^*_i$, any bundle that intersects $\{o \in O \mid \min_{P_i}(\omega_i) \mathrel{P_i} o\}$ is worse than any bundle that does not. Clearly, any \IR rule on $\mathcal{P}^*$ also satisfies \WELB. Furthermore, if preferences $P_i$ are assumed to be lexicographic (or responsive) when restricted to ``acceptable'' bundles (subsets of $\{o \in O \mid o \mathrel{R_i} \min_{P_i}(\omega_i)\}$), then our characterization results go through unchanged on $\mathcal{P}^*$.} We do not impose this additional restriction; instead we treat \IR and \WELB as independent properties in order to separate their respective roles in the results.

\subsubsection{Incentive properties}\label{incentive properties}

Given a preference profile $P = (P_i)_{i \in N}$ and a preference relation $P'_{i}$ for agent~$i$, we denote by $\left(P'_{i},P_{-i}\right)$ the preference profile in which agent~$i$ reports $P'_i$ and every other agent~$j \in N \setminus \{i\}$ reports $P_{j}$. Given a rule $\varphi$, we say that agent $i$ can \defterm{manipulate} $\varphi$ at $P$ by reporting $P'_i$ if $\varphi_i(P'_i, P_{-i}) \mathrel{P_i} \varphi_i(P)$. A rule is \defprop{strategy-proof} if no agent can manipulate it by reporting any preference relation.

\begin{definition}\label{def: strategy-proofness}
    A rule $\varphi$ satisfies \defprop{strategy-proofness (\SP)} if, for each $P \in \mathcal{P}$, each $i\in N$, and each $P'_i \in \mathcal{P}_i$, $\varphi_i(P) \mathrel{R_i} \varphi_i(P'_i,P_{-i})$.
\end{definition}

Because \PE, \IR, and \SP are incompatible once some agent owns more than one object (\citet{sonmez1999}, \citet{todo2014}), we will not insist on full \SP. Instead, we relax \SP by restricting the set of manipulation strategies under consideration. In particular, we focus on simple manipulations---subset-drop, drop, and truncation strategies---that act directly on agents' marginal preferences over individual objects and are especially salient when studying marginal rules (see, e.g., \citet{biro2022,biro2022serial}, \citet{altuntacs2023}).

\paragraph{Subset-drop, drop, and truncation strategies.}
Fix an agent~$i$ and a preference relation $P_i \in \mathcal{P}_i$. A subset-drop strategy for $P_i$ is a report whose marginal preference is obtained from $P_i|_O$ by ``dropping'' some subset $X \subseteq O \setminus \omega_i$ to the bottom, preserving the relative order within $X$ and within $O \setminus X$. Any preference $P'_i \in \mathcal{P}_i$ with this marginal preference is called a subset-drop strategy for $P_i$. Formally, $P'_i \in \mathcal{P}_i$ is a \defterm{subset-drop strategy} for $P_i$ if there exists $X \subseteq O \setminus \omega_i$ such that 
\begin{enumerate}
    \item[(i)] $P_i |_X = P'_i |_X$, 
    \item[(ii)] $P_i |_{O \setminus X} = P'_i |_{O \setminus X}$, and 
    \item[(iii)] for all $x \in X$ and all $y \in O \setminus X$, $y \mathrel{P'_i} x$. 
\end{enumerate}    
In this case, we say that $P'_i$ (respectively $P'_i |_O$) is obtained from $P_i$ (respectively $P_i |_O$) by \defterm{dropping} $X$. Let $\mathcal{S}_i(P_i)$ denote the set of all subset-drop strategies for $P_i$.

Within the class of subset-drop strategies, we distinguish two important subclasses. A subset-drop strategy $P'_i$ is called a \defterm{drop strategy} for $P_i$ if it is obtained by dropping a singleton subset $\{x\} \subseteq O \setminus \omega_i$, i.e., by moving a single object $x \in O \setminus \omega_i$ to the bottom of the marginal preferences. Let $\mathcal{D}_i(P_i)$ denote the set of all drop strategies for $P_i$.

A subset-drop strategy $P'_i$ is called a \defterm{truncation strategy} for $P_i$ if it is obtained by dropping a \defterm{tail subset} of the form $\{o \in O \setminus \omega_i \mid x \mathrel{R_i} o\}$ for some $x \in O$, i.e., by moving all objects that are weakly worse than some object~$x$ to the bottom of the marginal preferences. Let $\mathcal{T}_i(P_i)$ denote the set of all truncation strategies for $P_i$.

It is convenient to single out those truncation strategies with a specified cutoff object. A truncation strategy $P'_i$ is called a \defterm{truncation of $P_i$ at $y$} if it is obtained by dropping the strict tail subset $\{o \in O \setminus \omega_i \mid y \mathrel{P_i} o\}$.\footnote{\label{alternative truncation}Equivalently, $P'_i$ is a truncation of $P_i$ at $y$ if (i) $P_i |_{\omega_i} = P'_i |_{\omega_i}$, (ii) $P_i |_{O \setminus \omega_i} = P'_i |_{O \setminus \omega_i}$, and (iii) for each $x \in O \setminus \omega_i$ with $y \mathrel{P_i} x$, $\min_{P'_i}(\omega_i) \mathrel{P'_i} x$.} In this case, there is no $z \in O \setminus \omega_i$ such that $y \mathrel{P'_i} z \mathrel{P'_i} \min_{P_i}(\omega_i)$.

In the single-object environment (Section~\ref{sec:SS}), where $\omega_i$ is a singleton and thus $\omega_i = \min_{P_i}(\omega_i)$, a subset-drop strategy can be interpreted as ``declaring some of the other agents' objects unacceptable'', and a truncation strategy simply ``shortens the list of acceptable objects'' in the standard sense from two-sided matching (see, e.g., \citet{mongell1991sorority,roth1999truncation,ehlers2008}). This intuition extends to the multi-object environment, though objects ranked below $\min_{P_i}(\omega_i)$ need not be unacceptable in the usual sense.\footnote{Our results continue to hold on related domains in which bundles containing objects ranked below $\min_{P_i}(\omega_i)$ are explicitly unacceptable; see footnote~\ref{footnote other domain} and the surrounding discussion.} Nevertheless, under \WELB, a subset-drop strategy effectively ``vetoes'' some of the other agents' objects: no rule satisfying \WELB ever assigns to agent~$i$ any object ranked below $\min_{P_i}(\omega_i)$.

\begin{definition}
\label{def: truncation-proofness}
    A rule $\varphi$ satisfies 
    \begin{itemize}
        \item \defprop{truncation-proofness (\TP)} if, for each $P\in \mathcal{P}$, each $i\in N$, and each $P'_i \in \mathcal{T}_i(P_{i})$, we have $\varphi_i(P)\mathrel{R_i}\varphi_i(P'_i,P_{-i})$.
        \item \defprop{drop strategy-proofness (\DSP)} if, for each $P \in \mathcal{P}$, each $i\in N$, and each $P'_i \in \mathcal{D}_i(P_{i})$, we have $\varphi_i(P)\mathrel{R_i}\varphi_i(P'_i,P_{-i})$.
        \item \defprop{subset-drop strategy-proofness (\SDSP)} if, for each $P \in \mathcal{P}$, each $i \in N$, and each $P'_i \in \mathcal{S}_i(P_i)$, we have $\varphi_i(P) \mathrel{R_i} \varphi_i(P'_i, P_{-i})$.
    \end{itemize}
\end{definition}

\SDSP implies both \TP and \DSP, as it defends against the largest set of manipulation strategies. \TP and \DSP are logically independent in general.

\begin{example}
\label{example: heuristics}
Let $O = \{a, b, c, d, e, x, y\}$, $\omega_{i}=\{x, y\}$, and $P_i \in \mathcal{P}_i$ be such that $P_i |_O:a,b,\underline{x},c,d,\underline{y},e$ (agent~$i$'s endowment is underlined for emphasis). Then:
    \begin{itemize}
    \item any $P^1_i \in \mathcal{P}_i$ with $P^1_i |_O:b,\underline{x},c,d,\underline{y},e,a$ is a drop strategy for $P_i$, obtained by dropping object $a$. Note that $P^1_i$ is not a truncation strategy for $P_i$ (e.g., the relations disagree about the best object in $O \setminus \omega_i$);
    \item any $P^2_i, P^3_i \in \mathcal{P}_i$ with $P^2_i |_O:a,b,\underline{x},c,\underline{y},d,e$ and $P^3_i |_O:a,\underline{x},\underline{y},b,c,d,e,$ are truncation strategies for $P_i$, obtained by dropping the tail subsets $\{d,e\}$ and $\{b,c,d,e\}$, respectively. Neither $P^2_i$ nor $P^3_i$ is a drop strategy for $P_i$ (however, each can be obtained from $P_i$ via a sequence of drop strategies). \hfill $\diamond$
    \end{itemize}
\end{example}

\begin{remark}\label{drop strategies can change order}
    Fix $i \in N$ and $P_i \in \mathcal{P}_i$. For any object $x \in O \setminus \omega_i$, there exists a subset-drop strategy for $P_i$ that ranks $x$ at the top of agent~$i$'s marginal preferences over $O \setminus \omega_i$. The same reported preferences can be obtained by dropping, one at a time, the objects in $O \setminus \omega_i$ that agent~$i$ ranks above $x$. 
    This observation simplifies proofs of characterizations of TTC based on \DSP and \SDSP (and on \SP for the single-object environment). By contrast, truncation strategies preserve the ranking of objects in $O \setminus \omega_i$ (see footnote~\ref{alternative truncation}), so they do not permit this kind of ``push-up'' manipulation.
\end{remark}

\subsection{Top Trading Cycles}\label{Subsection: Definition of generalized TTC}

A trading cycle at $\omega$ (or simply a \defterm{trading cycle}\footnote{That is, a \defterm{trading cycle} refers to a trading cycle at the particular allocation $\mu = \omega$ (as defined in Section~\ref{sec:Efficiency}).}) is a cyclic sequence
\[
C=\left(o_1, i_{1},o_2,i_{2},\dots,i_{k-1},o_k,i_{k},o_{k+1}=o_1\right)
\]
consisting of $k \geq 1$ distinct objects and $k$ distinct agents such that, for each $\ell \in \{1, \dots, k\}$, $o_\ell \in \omega_{i_\ell}$. Each object $o_\ell$ on $C$ precedes (or ``points to'') its owner $i_\ell$, and each agent $i_\ell$ on $C$ precedes (or ``points to'') the object $o_{\ell+1}$. The sets of agents and objects on $C$ are denoted
by $N\left(C\right)=\left\{i_{1},i_{2},\dots,i_{k}\right\}$ and $O(C)=\left\{o_1,o_2,\dots,o_k\right\}$, respectively. An allocation $\mu$ is said to \defterm{execute} the trading cycle $C$ if it assigns to each agent in $N\left(C\right)$ the object she points to on $C$; that is, for each $i_{\ell}\in N\left(C\right)$, $o_{\ell+1} \in \mu_{i_{\ell}}$.

We study a multi-object extension of Gale's Top Trading Cycles algorithm (\citet{shapley1974}), which determines an allocation by executing a sequence of carefully chosen trading cycles at $\omega$. At each step of this procedure, every agent points to her most-preferred unassigned object, and every unassigned object points to its owner. The resulting directed graph contains at least one trading cycle, and each agent involved in a trading cycle is assigned the object to which she points. All objects involved in a trading cycle are then removed. If unassigned objects remain, then the procedure continues to the next step; otherwise, it terminates with the resulting allocation.

We now formalize the \defterm{(generalized) Top Trading Cycles rule (TTC)}, which we denote by $\varphi^{\text{TTC}}$. Given a preference profile $P \in \mathcal{P}$, the allocation $\varphi^{\text{TTC}}(P)$ is determined by running the following \defterm{(generalized) TTC algorithm} at $P$. We denote this specific instance as $\text{TTC}(P)$.

\noindent \rule{1\columnwidth}{1pt}
\paragraph{Algorithm:}
$\text{TTC}\left(P\right)$
\begin{description}
\item [\textnormal{\emph{Initialization}:}] 
Set $\mu^{0}\coloneqq\left(\emptyset\right)_{i\in N}$ and $O^{1}\coloneqq O$.
\item [\textnormal{\emph{Step~$t \geq 1$}:}] ~
    Each agent~$i$ points to $\max_{P_i}(O^t)$, and each object in $O^t$ points to its owner. Let $\mathcal{C}_t(P)$ be the resulting set of trading cycles. For each $C \in \mathcal{C}_t(P)$, assign to each $i \in N(C)$ the object she points to in $C$; this yields the partial allocation $\mu^t$ with 
    $$\mu^t_i =
    \begin{cases}
        \mu_i^{t-1} \cup \{\max_{P_i}(O^t)\}, & \text{if } i \in \bigcup_{C \in \mathcal{C}_t(P)} N(C), \\
        \mu_i^{t-1}, & \text{otherwise}.
    \end{cases}
    $$
    Remove all objects involved in trading cycles to form the set $O^{t+1} = O^t \setminus \bigcup_{C \in \mathcal{C}_t(P)}O(C)$. Proceed to step~$t+1$ if $O^{t+1} \neq \emptyset$; proceed to Termination otherwise.
\item [\textnormal{\emph{Termination}:}] Because at least one object is removed at each step,
the algorithm terminates at some step $T$. Return the allocation
$\varphi^{\text{TTC}}\left(P\right)\coloneqq\mu^{T}$.
\end{description}
\rule{1\columnwidth}{1pt}

Example~\ref{TTC example} illustrates the TTC algorithm with a simple three-agent problem.

\begin{example}\label{TTC example}
    Suppose $N = \{1,2,3\}$, $O = \{a,b,c,d\}$, and $\omega = (\{a,b\}, \{c\}, \{d\})$. Consider a preference profile $P$ with marginal preferences $P_1 |_O:c,a,d,b$, $P_2 |_O:a,b,c,d$, and $P_3 |_O:a,c,b,d$. The algorithm $\text{TTC}(P)$ operates as follows.
    \begin{description}
        \item [\textnormal{\emph{Step~$1$}:}] ~ Each agent points to her most-preferred object in $O$: agent~$1$ points to $c$, while agents~$2$ and~$3$ point to $a$. Each object in $O$ points to its owner. There is a trading cycle $C_1(P) = (c,2,a,1,c)$, which is executed to yield partial allocation $\mu^1 = (\{c\},\{a\},\emptyset)$ and the remaining objects $O^2 = \{b,d\}$.
        \item [\textnormal{\emph{Step~$2$}:}] ~ Each agent points to her most-preferred object in $O^2 = \{b,d\}$: agent~$1$ points to $d$, while agents~$2$ and~$3$ point to $b$. Each object in $O^2$ points to its owner. There is a trading cycle $C_2(P) = (d,3,b,1,d)$, which is executed to yield the final allocation $\mu^2 = (\{c,d\}, \{a\},\{b\})$ and the remaining objects $O^3 = \emptyset$.
    \end{description}
    The algorithm therefore terminates with the allocation $\varphi^{\text{TTC}}(P) = (\{c,d\}, \{a\}, \{b\})$. \hfill $\diamond$
\end{example}

The following fact is immediate from the definition of $\varphi^{\mathrm{TTC}}$.

\begin{fact}\label{TTCproperties}
    On any domain of strict preferences, $\varphi^{\text{TTC}}$ is marginal and satisfies \BAL and \WELB.
\end{fact}

Beyond these basic properties, the behavior of $\varphi^{\text{TTC}}$ depends on the
preference domain.

\section{Lexicographic preferences}
\label{results for lexicographic preferences}

The lexicographic domain $\mathcal{L}$ is a natural starting point for our analysis because of its tractability. As discussed in Sections~\ref{preference domains}--\ref{sec:Properties}, this domain has two convenient features. First, each $P_i \in \mathcal{L}_i$ is uniquely determined by its marginal preference $P_i|_O$, so any rule defined on $\mathcal{L}$ is automatically marginal. Second, on $\mathcal{L}$, \PE coincides with \IGE.

On the lexicographic domain, $\varphi^{\mathrm{TTC}}$ behaves especially well. \citet{fujita2018} show that it always selects an allocation in the core of the associated exchange economy, and hence it satisfies both \PE and \IR. Although $\varphi^{\mathrm{TTC}}$ necessarily fails \SP,\footnote{\PE, \IR, and \SP are incompatible even for lexicographic preferences (\citet{todo2014}).} \citet{altuntacs2023} prove that it satisfies \SDSP; thus it also satisfies \TP and \DSP.\footnote{\citet{altuntacs2023} refer to \SDSP as ``subset total drop strategy-proofness''. Their notion of drop strategy-proofness is stronger than ours, as it also precludes manipulations obtained by dropping objects an agent already owns. Throughout, we use the weaker \DSP as it suffices for our characterization results.} These observations yield the following fact.

\begin{fact}
    On the lexicographic domain, $\varphi^{\text{TTC}}$ satisfies \PE, \IR, and \SDSP.\label{necessity}
\end{fact}

Facts~\ref{TTCproperties} and~\ref{necessity} together imply that  $\varphi^{\mathrm{TTC}}$ satisfies \BAL, \IGE, \WELB, and \TP. Theorem~\ref{Multi-unit characterization with TP} shows that these four properties characterize $\varphi^{\mathrm{TTC}}$ on $\mathcal{L}$.

\begin{theorem}
\label{Multi-unit characterization with TP}
On the lexicographic domain, a rule satisfies \BAL, \IGE, \WELB, and \TP if and only if it equals $\varphi^{\text{TTC}}$. 
\end{theorem}

The proof proceeds by minimal counterexample, using a novel minimality criterion that combines a ``size'' function (\citet{sethuraman2016}) and a ``similarity'' function (\citet{ekici2022}). Suppose $\varphi$ satisfies the properties but differs from $\varphi^{\text{TTC}}$. Call $P \in \mathcal{L}$ a \defterm{conflict profile} if $\varphi(P) \neq \varphi^{\text{TTC}}(P)$. For each conflict profile $P$, define its similarity $\rho(P)$ as the earliest step $t$ of $\text{TTC}(P)$ such that $\varphi(P)$ does not execute all trading cycles in $\mathcal{C}_t(P)$. Among all conflict profiles that minimize $\rho$, choose a profile that further minimizes the size function $s(P) \coloneqq \sum_{i \in N}|\{o \in O \mid o \mathrel{R_i} \min_{P_i}(\omega_i)\}|$. 
At the chosen conflict profile $P$, the deviation from $\varphi^{\mathrm{TTC}}(P)$ forces a Pareto-improving trading cycle at $\varphi(P)$, contradicting \IGE (see Appendix~\ref{appendix:proof}).

Theorem~\ref{Multi-unit characterization with TP} characterizes $\varphi^{\mathrm{TTC}}$ using \TP as its incentive requirement. Our next theorem shows that \TP can be replaced with the alternative requirement \DSP. The key observation is that \DSP becomes strictly stronger once \WELB is imposed: on the lexicographic domain, \DSP together with \WELB implies \SDSP, and hence \TP.\footnote{\TP is not implied by \DSP alone. Conversely, \DSP is not implied by any proper subset of the properties in Theorem~\ref{Multi-unit characterization with TP}.}

\begin{lemma}
    \label{Proposition: drop SP + WELB implies subset-drop SP}
    On the lexicographic domain, if a rule satisfies \DSP and \WELB, then it satisfies \SDSP.
\end{lemma}

Theorem~\ref{Multi-unit characterization with drop SP} refines two characterizations provided by \citet[Theorems~1 and~3]{altuntacs2023}.

\begin{theorem}
\label{Multi-unit characterization with drop SP}
On the lexicographic domain, a rule satisfies \BAL, \IGE, \WELB, and \DSP if and only if it equals $\varphi^{\text{TTC}}$.
\end{theorem}

\begin{proof}
By Facts~\ref{TTCproperties} and~\ref{necessity}, $\varphi^{\mathrm{TTC}}$ satisfies the stated properties. Conversely, if $\varphi$ is a rule satisfying \BAL, \IGE, \WELB, and \DSP, then Lemma~\ref{Proposition: drop SP + WELB implies subset-drop SP} implies that $\varphi$ satisfies \SDSP and hence \TP. Therefore, $\varphi$ satisfies the properties in Theorem~\ref{Multi-unit characterization with TP}, which yields $\varphi = \varphi^{\mathrm{TTC}}$. 
\end{proof}

Note that \IR cannot replace \WELB in the statement of Theorem~\ref{Multi-unit characterization with TP} or~\ref{Multi-unit characterization with drop SP}.\footnote{Of course, \IR can be \emph{added} to the list of properties but it is \emph{implied} by the four properties in each theorem.} Example~\ref{no characterization with IR} (Appendix~\ref{Appendix: independence}) exhibits a rule $\varphi^{\neg\mathrm{WELB}}$ on $\mathcal{L}$ that differs from $\varphi^{\mathrm{TTC}}$ yet also satisfies \BAL, \PE, \IR, \TP, and \DSP. The construction relies on a key distinction between the two participation guarantees: \IR allows an agent to receive an individually unattractive object as part of a sufficiently desirable bundle, whereas \WELB prevents such objects from appearing in her assignment. Interestingly, the analogous construction is not possible with marginal rules on the responsive domain; in that setting, \IR suffices for the characterization (see Lemma~\ref{individual-good + IR implies WELB} and Theorem~\ref{responsive characterization with IR}).\footnote{Similarly, this construction is not possible on the domain $\mathcal{P}^*$ described in footnote~\ref{footnote other domain}, as \IR implies \WELB on that domain.}

More generally, Appendix~\ref{Appendix: independence} provides several alternative rules that demonstrate the independence of our properties. Table~\ref{tab:properties_L} summarizes the properties satisfied by these rules on the lexicographic domain; it shows that each of the properties in Theorems~\ref{Multi-unit characterization with TP} and~\ref{Multi-unit characterization with drop SP} is indispensable.

\begin{table}[ht]
\caption{Properties of selected rules on $\mathcal{L}$}
\label{tab:properties_L}
\centering
\renewcommand{\arraystretch}{1.15}
\begin{tabular}{l*{7}{c}}
\hline
Rule & \BAL & \PE & \WELB & \TP & \DSP & \IR & \SP \\
\midrule
$\varphi^{\mathrm{TTC}}$                     & \checkmark & \checkmark & \checkmark & \checkmark & \checkmark & \checkmark &            \\
No-trade rule (Ex.~\ref{No-trade rule}) & \checkmark &            & \checkmark & \checkmark & \checkmark & \checkmark & \checkmark \\
Balanced Serial Dictatorship (Ex.~\ref{Balanced SD})  & \checkmark & \checkmark &            & \checkmark & \checkmark &            & \checkmark \\
$\varphi^{\neg \TP}$ (Ex.~\ref{Deviation from TTC})                     & \checkmark & \checkmark & \checkmark &            &            & \checkmark &            \\
$\varphi^{\neg \BAL}$ (Ex.~\ref{NOT BAL}) &            & \checkmark & \checkmark & \checkmark & \checkmark &            &            \\
$\varphi^{\neg \mathrm{WELB}}$ (Ex.~\ref{no characterization with IR})                     & \checkmark & \checkmark &   & \checkmark & \checkmark  & \checkmark &            \\
\bottomrule
\end{tabular}

\end{table}

\section{Responsive preferences}\label{sec:responsivedomain}

We now turn from lexicographic to responsive preferences. Objects remain substitutes, but the tension between efficiency, individual rationality, and incentive compatibility becomes even starker on this more general domain.

On the responsive domain, \PE becomes a demanding benchmark: Pareto improvements may require coordinating intricate multi-object exchanges among several agents, and even deciding whether a Pareto improvement exists is NP-complete (\citet{de2009complexity}, \citet{aziz2019efficient}). Furthermore, \citet{manjunath2024} show that no marginal rule can simultaneously satisfy \PE and \IR.
Their argument is based on the following simple two-agent example.

\begin{example}
    \label{westkamp}
    Let $N = \{1, 2\}$, $O = \{a, b, c, d\}$, and $\omega = (\{a, d\}, \{b,c\})$. Let $\varphi$ be a rule on $\mathcal{R}$ satisfying \PE and \IR.
    
    Consider $P \in \mathcal{L}$ with $P_1 |_O =  P_2 |_O:a,b,c,d$. Since $\varphi$ satisfies \PE and \IR, it must assign $\varphi(P) = \omega$. Now consider $P' \in \mathcal{R}$ with $P' |_O = P |_O$, and assume that each agent prefers the other's endowment: $\{b,c\} \mathrel{P'_1} \{a,d\}$ and $\{a,d\} \mathrel{P'_2} \{b,c\}$. 
    Then $\omega$ is Pareto dominated at $P'$ by the bundle swap $\mu \coloneqq (\{b, c\}, \{a, d\})$.
    Since $\varphi$ satisfies \PE, we must have $\varphi(P') \neq \omega = \varphi(P)$. Thus $P' |_O = P |_O$ and $\varphi(P') \neq \varphi(P)$, so $\varphi$ is not marginal.  
    \hfill $\diamond$
\end{example}

At the lexicographic profile $P$ in Example~\ref{westkamp}, any rule on $\mathcal{R}$ that satisfies \PE and \IR when restricted to $\mathcal{L}$ must select the endowment $\omega = (\{a,d\}, \{b,c\})$. If the rule is marginal, then it must also select $\omega$ at the responsive profile $P'$ with the same marginal preferences. In particular,  $\varphi^{\mathrm{TTC}}$ assigns $\omega$ at both $P$ and $P'$, and thus $\varphi^{\mathrm{TTC}}(P') = \omega$ is Pareto dominated at $P'$. However, no trading cycle is Pareto improving at $P'$ (under the common marginal preferences, a single-object exchange would require some agent to trade away an object for a lower-ranked one).

More generally, $\varphi^{\mathrm{TTC}}$ inherits individual-good efficiency on $\mathcal{R}$ from its Pareto efficiency on $\mathcal{L}$: if a Pareto-improving trading cycle existed at $\varphi^{\mathrm{TTC}}(P)$ for some $P \in \mathcal{R}$, then the same cycle would be Pareto improving at the lexicographic profile with the same marginal preferences, contradicting the Pareto efficiency of $\varphi^{\mathrm{TTC}}$ on $\mathcal{L}$. In addition, \citet{altuntacs2023} show that $\mathcal{R}$ is a maximal domain on which $\varphi^{\mathrm{TTC}}$ satisfies \IR. The next fact collects these properties.

\begin{fact}\label{TTC properties on R}
On the responsive domain, $\varphi^{\mathrm{TTC}}$ satisfies \IGE and \IR.
\end{fact}

Incentive compatibility becomes even more fragile on the responsive domain because extending a rule to a larger domain preserves any manipulations that were already possible and may introduce new ones. Nevertheless, $\varphi^{\mathrm{TTC}}$ remains robust against truncation strategies on $\mathcal{R}$.

\begin{proposition}\label{prop:possibility}
    On the responsive domain, $\varphi^{\text{TTC}}$ satisfies TP.
\end{proposition}

Combining Fact~\ref{TTCproperties}, Fact~\ref{TTC properties on R}, and Proposition~\ref{prop:possibility}, we see that $\varphi^{\mathrm{TTC}}$ satisfies \BAL, \IGE, \WELB, and \TP on $\mathcal{R}$. Within the class of marginal rules, these properties continue to identify $\varphi^{\mathrm{TTC}}$ on the larger domain $\mathcal{R}$ (cf. Theorem~\ref{Multi-unit characterization with TP}). In fact, in Theorem~\ref{responsive characterization} below we show that the uniqueness holds more generally. The key observation is the following lemma, which says that \BAL is superfluous: for marginal rules on $\mathcal{R}$, it is implied by \WELB and \TP.

\begin{lemma}\label{lemma:BAL implied on responsive domain}
    On the responsive domain, if a marginal rule satisfies \WELB and \TP, then it satisfies \BAL.
\end{lemma}

\begin{theorem}\label{responsive characterization}
    On the responsive domain, a marginal rule satisfies \IGE, \WELB, and \TP if and only if it equals $\varphi^{\text{TTC}}$.
\end{theorem}

\begin{proof}
We have shown that $\varphi^{\mathrm{TTC}}$ satisfies the stated properties. 
Conversely, let $\varphi$ be a marginal rule on $\mathcal{R}$ satisfying \IGE, \WELB, and \TP. By Lemma~\ref{lemma:BAL implied on responsive domain}, $\varphi$ also satisfies \BAL. Hence Theorem~\ref{Multi-unit characterization with TP} implies that $\varphi$ coincides with $\varphi^{\text{TTC}}$ on $\mathcal{L}$. 

Fix any $P' \in \mathcal{R}$, and let $P$ be the unique lexicographic preference profile with $P|_O = P'|_O$.  Since $\varphi$ and $\varphi^{\text{TTC}}$ are marginal, we have
$$\varphi(P') = \varphi(P) = \varphi^{\text{TTC}}(P) = \varphi^{\text{TTC}}(P').$$
Thus $\varphi$ coincides with $\varphi^{\text{TTC}}$ on $\mathcal{R}$.\footnote{The argument parallels one used by \citet{biro2022} to extend a characterization of the ``Circulation Top Trading Cycle'' rule from the lexicographic to the responsive domain in a model with homogeneous objects.}
\end{proof}

Unlike on the lexicographic domain, where \IR cannot generally replace \WELB (see Example~\ref{no characterization with IR}), on $\mathcal{R}$ any marginal rule satisfying \IR also satisfies \WELB (Lemma~\ref{individual-good + IR implies WELB}). Thus \WELB can be replaced by \IR in Theorem~\ref{responsive characterization}.

\begin{theorem}\label{responsive characterization with IR}
    On the responsive domain, a marginal rule satisfies
    \IGE, \IR, and \TP
    if and only if it equals $\varphi^{\text{TTC}}$.
\end{theorem}

\begin{proof}
We have shown that $\varphi^{\mathrm{TTC}}$ satisfies the stated properties. Conversely, if $\varphi$ is a marginal rule satisfying \IGE, \IR, and \TP, then Lemma~\ref{individual-good + IR implies WELB} implies that $\varphi$ also satisfies \WELB. Therefore, $\varphi$ satisfies the properties in Theorem~\ref{responsive characterization}, and hence $\varphi = \varphi^{\mathrm{TTC}}$.
\end{proof}

Theorems~\ref{responsive characterization} and~\ref{responsive characterization with IR} characterize $\varphi^{\mathrm{TTC}}$ on $\mathcal{R}$ using \TP as the incentive requirement. The next example shows that $\varphi^{\mathrm{TTC}}$ fails the alternative requirement \DSP on $\mathcal{R}$.

\begin{example}
    \label{example:Not drop SP on additive domain}    
Let $N = \{1, 2, 3 \}$ and $\omega=\left(\left\{ a,b\right\} ,\left\{ c\right\} ,\left\{ d\right\} \right)$. Let $P_2$ and $P_3$ be lexicographic preferences with $P_2|_O : d,b,c,a$ and $P_3|_O:b,a,c,d$.

Let $P_1$ be a responsive preference with $P_1|_O : d,b,c,a$ and $\{b,c\} \mathrel{P_1} \{a,d\}$. Then $\varphi^{\mathrm{TTC}}$ assigns $\varphi^{\mathrm{TTC}}(P) = ( \{a,d\}, \{c\}, \{b\})$. Now suppose agent~$1$ drops object $d$, reporting some $P'_1$ with the marginal preference $P'_1|_O : b,c,a,d$. Then $\varphi^{\mathrm{TTC}}$ assigns $\varphi^{\mathrm{TTC}}(P'_1, P_{-1}) = (\{b,c\}, \{d\}, \{a\} )$. Because $\{b,c\} \mathrel{P_1} \{a,d\}$, agent~$1$ benefits from the drop strategy.
\hfill $\diamond$
\end{example}

Example~\ref{example:Not drop SP on additive domain} suggests that \DSP is a rather restrictive requirement on $\mathcal{R}$. Combined with our lexicographic-domain characterization (Theorem~\ref{Multi-unit characterization with drop SP}) and the fact that \IR implies \WELB for marginal rules on $\mathcal{R}$ (Lemma~\ref{individual-good + IR implies WELB}), this observation yields a sharp impossibility.

\begin{corollary}\label{impossibility1}
On the responsive domain, no marginal rule satisfies \IGE, \IR, and \DSP.
\end{corollary}
\begin{proof}
    Toward contradiction, let $\varphi$ be a marginal rule on $\mathcal{R}$ satisfying \IGE, \IR, and \DSP. Then $\varphi$ also satisfies \BAL and \WELB by Lemma~\ref{individual-good + IR implies WELB}. Hence Theorem~\ref{Multi-unit characterization with drop SP} implies that $\varphi$ coincides with $\varphi^{\mathrm{TTC}}$ on $\mathcal{L}$, and by marginality it must coincide with $\varphi^{\mathrm{TTC}}$ on $\mathcal{R}$. Thus $\varphi$ fails \DSP, a contradiction. 
\end{proof}

Corollary~\ref{impossibility1} shows that any marginal rule on $\mathcal{R}$ that guarantees \IGE and \IR is necessarily manipulable---even among drop strategies.\footnote{A similar argument yields an analogous impossibility result with \IR replaced by \BAL and \WELB: no marginal rule on $\mathcal{R}$ satisfies \IGE, \BAL, \WELB, and \DSP.}
Nevertheless, we show that $\varphi^{\textrm{TTC}}$ retains some appealing incentive qualities beyond \TP. In particular, it is \defterm{not obviously manipulable} according to the typology proposed by \citet{troyan2020}. Roughly speaking, although profitable deviations exist, no deviation is clearly beneficial when the agent compares only the best- and worst-case outcomes that can result from truthful reporting and from the deviation. 

For each $i \in N$ and each $P_i \in \mathcal{R}_i$, agent~$i$'s \defterm{opportunity set} at $P_i$ under $\varphi^{\mathrm{TTC}}$ is the set $\mathcal{O}_i(P_i) = \{\varphi_i^{\mathrm{TTC}}(P_i,P^*_{-i}) \mid P^*_{-i} \in \mathcal{P}_{-i}\}$ of all bundles that $i$ could be assigned upon reporting $P_i$. For each nonempty subset $\mathcal{X} \subseteq 2^O$, $B_{P_i}(\mathcal{X})$ denotes the most-preferred bundle in $\mathcal{X}$ according to $P_i$, i.e., $B_{P_i}(\mathcal{X}) \in \mathcal{X}$ and $B_{P_i}(\mathcal{X}) \mathrel{R_i} X$ for each $X \in \mathcal{X}$. Similarly, $W_{P_i}(\mathcal{X})$ denotes the least-preferred bundle in $\mathcal{X}$ according to $P_i$.

\begin{proposition}\label{prop:NOM}
On the responsive domain, $\varphi^{\mathrm{TTC}}$ is  \defterm{not obviously manipulable}.\newline 
More precisely, for each $P \in \mathcal{R}$, each $i \in N$, and each $P'_i \in \mathcal{R}_i$ with $\varphi^{\mathrm{TTC}}_i(P'_i,P_{-i}) \mathrel{P_i} \varphi^{\mathrm{TTC}}_i(P)$,
\begin{itemize}
    \item[(i)] $\displaystyle B_{P_i} \left(
    \mathcal{O}_i(P_i)
    \right) \mathrel{R_i}
    B_{P_i} \left(
    \mathcal{O}_i(P'_i)
    \right)
    $, and
    \item[(ii)] $\displaystyle W_{P_i} \left(
    \mathcal{O}_i(P_i)
    \right) \mathrel{R_i} 
    W_{P_i} \left(
    \mathcal{O}_i(P'_i)
    \right).
    $
\end{itemize}
\end{proposition}

Table~\ref{tab:properties_R} summarizes the properties satisfied by several alternative rules on $\mathcal{R}$ (see the constructions in Appendix~\ref{Appendix: independence}). It shows that the properties used in Theorems~\ref{responsive characterization} and~\ref{responsive characterization with IR}, including marginality, are logically independent.

\begin{table}[ht]
\caption{Properties of selected rules on $\mathcal{R}$}
\label{tab:properties_R}
\centering
\renewcommand{\arraystretch}{1.15}
\begin{tabular}{lccccccc}
\toprule
Rule & \BAL & \IGE & \WELB & \TP & \DSP & \IR & \MAR \\
\midrule
$\varphi^{\mathrm{TTC}}$                     & \checkmark & \checkmark & \checkmark & \checkmark &            & \checkmark & \checkmark \\
No-trade rule (Ex.~\ref{No-trade rule}) & \checkmark &            & \checkmark & \checkmark & \checkmark & \checkmark & \checkmark \\
Balanced Serial Dictatorship (Ex.~\ref{Balanced SD})  & \checkmark & \checkmark &            & \checkmark & \checkmark &            & \checkmark \\
$\varphi^{\neg \TP}$ (Ex.~\ref{Deviation from TTC})                     & \checkmark & \checkmark & \checkmark &            &            & \checkmark & \checkmark           \\
$\varphi^{\neg \MAR}$ (Ex.~\ref{non-marginal})                     & \checkmark & \checkmark & \checkmark & \checkmark          &            & \checkmark &            \\
\bottomrule
\end{tabular}
\end{table}

\section{The single-object environment}     
\label{sec:SS}

We now specialize our analysis to the classic single-object environment (the ``housing market'') introduced by \citet{shapley1974}. In this environment our properties simplify, and the multi-object characterization in Theorem~\ref{Multi-unit characterization with TP} collapses to a sharper statement. 

We formalize the single-object environment as follows. Assume that $O = \{o_1, \dots, o_n\}$ and that each agent $i \in N$ is endowed with the object $\omega_i = o_i$. Each agent~$i$ has strict preferences $P_i$ over individual objects in $O$, and $\mathcal{P} = \prod_{i \in N}\mathcal{P}_i$ denotes the domain of strict preference profiles.\footnote{In this section only, we adopt the standard convention that each agent has preferences $P_i$ defined directly on the set of objects $O$. This is only a notational simplification. If instead we maintained the assumption that each agent has lexicographic preferences on $2^O$, then the result below would be an immediate corollary of Theorem~\ref{Multi-unit characterization with TP}.} An allocation is represented as a bijection $\mu : N \to O$, where $\mu_i$ denotes the object assigned to agent~$i$. 

In the single-object environment, every allocation is balanced by definition. Moreover, \IR and \WELB coincide because each agent must be assigned exactly one object, and \PE and \IGE coincide because any Pareto improvement can be realized by executing a collection of disjoint trading cycles. Truncation and drop strategies---as well as the corresponding incentive properties, \TP and \DSP---are defined exactly as in Section~\ref{incentive properties}. In this setting, truncation strategies admit a familiar interpretation from the matching literature: a truncation strategy for $P_i$ simply shortens agent~$i$'s list of acceptable objects by moving the ``outside option'' $\omega_i$ up in the ranking while preserving the ordering of objects in $O \setminus \omega_i$ (see footnote~\ref{alternative truncation} and the surrounding discussion).

\begin{theorem}
    \label{Theorem: Single-unit characterization}
    In the single-object environment, a rule satisfies \PE, \IR, and \TP if and only if it equals $\varphi^{\mathrm{TTC}}$. 
    \end{theorem}

The proof closely parallels that of Theorem~\ref{Multi-unit characterization with TP} and is therefore omitted.

Theorem~\ref{Theorem: Single-unit characterization} refines several characterizations of $\varphi^{\mathrm{TTC}}$ in housing markets. In particular, \citet{ma1994} shows that $\varphi^{\mathrm{TTC}}$ is the unique rule satisfying \PE, \IR, and \SP, and \citet{altuntacs2023} prove that \SP can be replaced with either \DSP or the weaker ``upper invariance''.\footnote{Upper invariance requires an agent's assignment to remain unchanged when she misrepresents only the ordering of objects ranked below her assignment at the truthful profile.} In the single-object environment, upper invariance together with \IR implies \TP; hence Theorem~\ref{Theorem: Single-unit characterization} establishes the same uniqueness under weaker criteria.

The theorem also complements characterizations using relaxed efficiency requirements. For example, \citet{ekici2022} replaces \PE with ``pair efficiency'' while retaining \IR and \SP, and \cite{CHEN2025103190} further weaken \SP to upper invariance. Theorem~\ref{Theorem: Single-unit characterization} does not permit a similar refinement: pair efficiency together with \IR and \TP does not characterize $\varphi^{\mathrm{TTC}}$ (see \citet[Example~1]{CORENO2025112159}).

\section{Extension: Conditionally lexicographic preferences}\label{Results for CL preferences}
Conditionally lexicographic preferences generalize purely lexicographic preferences (\citet{booth2010learning}; see also \citet{domshlak2011,pigozzi2016}). Unlike responsive preferences, they allow the relative ranking of two objects to depend on the other objects they are obtained with: for example, an agent may prefer drinking Champagne to Bordeaux when paired with oysters, but Bordeaux to Champagne otherwise.
This flexibility accommodates complementarities among objects while retaining several appealing features of lexicographic preferences. In particular, conditionally lexicographic preferences admit a compact representation via \defterm{lexicographic preference trees (LP trees)}, which makes them attractive from an implementation perspective.\footnote{Intuitively, an LP tree is a rooted binary tree that represents conditional marginal preferences in a graphical manner. For any bundle $Y$, there is a unique root-to-leaf path that is consistent with $Y$. This path specifies the agent's preference ordering over objects conditional on receiving $Y$. See  Appendix~\ref{Appendix: LP trees} for details.} As we shall see, \PE coincides with  \IGE on this domain.

Loosely speaking, an agent has conditionally lexicographic preferences if, for any nonempty set of objects $X$ disjoint from $Y$, there is a unique object in $X$ which is the ``lexicographically best'' addition to $Y$ from $X$.

\begin{definition}\label{characterization of CL preferences}
    A preference relation $P_i$ on $2^O$ is \defterm{conditionally lexicographic} if,
    for all disjoint $X, Y \in 2^O$ with $X \neq \emptyset$, there is a unique object in $X$, denoted $x^* = \max_{P_i}(X \mid Y)$, such that
    \begin{equation*}
        \text{for each }Z \subseteq X \setminus \{x^*\} ,\quad (Y \cup x^* ) \mathrel{P_i} (Y \cup Z).
    \end{equation*}
    We call $x^* = \max_{P_i}(X \mid Y)$ agent~$i$'s most-preferred object in $X$ conditional on already having $Y$.
\end{definition}

For each $i \in N$, let $\mathcal{CL}_i$ be the set of conditionally lexicographic preferences on $2^O$, and write $\mathcal{CL} \coloneqq \prod_{i \in N}\mathcal{CL}_i$ for the \defterm{conditionally lexicographic domain}. Every $P_i \in \mathcal{CL}_i$ is monotonic. Moreover, every lexicographic preference is conditionally lexicographic, and a preference is lexicographic if and only if it is both conditionally lexicographic and responsive. Formally, for each $i \in N$, 
$\mathcal{L}_i \subseteq \mathcal{CL}_i \subseteq \mathcal{M}_i$ (with both inclusions strict if $|O| \geq 3$), and $\mathcal{L}_i = \mathcal{CL}_i \cap \mathcal{R}_i$.

\subsection{Properties: Conditional \WELB and \DSP}

In this section we look beyond marginal rules and consider rules that depend on agents' full preferences over bundles.
Among the properties we use below, the definitions of \BAL, \PE, \IGE, and \IR carry over from Section~\ref{sec:Properties} without modification. By contrast, \WELB and \DSP, which were originally defined with respect to agents' marginal preferences, must be reformulated using the conditional marginal preferences introduced below.\footnote{We do not discuss truncation-proofness on $\mathcal{CL}$ because the corresponding definition becomes rather unwieldy.}

Given $P_i \in \mathcal{CL}_i$ and a bundle $Y \subseteq O$, we define the \defterm{conditional marginal preference} $P_i(Y)|_O$ over individual objects \emph{conditional on receiving $Y$}. Formally, $P_i(Y)|_O$ is the strict linear order on $O$ such that, for all $x,y\in O$,\footnote{Equivalently, for all $x,y \in O$, $x \mathrel{P_i(Y)|_O} y \iff x\cup  (Y   \setminus \{x,y\}) \mathrel{P_i} y\cup  (Y   \setminus \{x,y\})$.}
$$
x \mathrel{P_i(Y)|_O} y \iff  (Y \cup \{x\}) \setminus \{y\} \mathrel{P_i} (Y \cup \{y\}) \setminus \{x\}.
$$

Equivalently, when $P_i$ is represented by an LP tree, $P_i(Y)|_O$ is the order in which the objects are encountered along the unique path consistent with~$Y$ (see Appendix~\ref{Appendix: LP trees}).
We define $\mathrel{R_i(Y)|_O}$ by $x \mathrel{R_i(Y)|_O} y$ if and only if ($x \mathrel{P_i(Y)|_O} y$ or $x = y$). For notational convenience, we will often write $P_i(Y)$ instead of $P_i(Y)|_O$ (and similarly $R_i(Y)$ instead of $R_i(Y)|_O$). Given a nonempty subset $X \subseteq O$, let $\max_{P_i(Y)}(X)$ and $\min_{P_i(Y)}(X)$ denote the most- and least-preferred objects in $X$ according to $P_i(Y)|_O$, respectively. Note that when $X$ and $Y$ are disjoint, $\max_{P_i(Y)}(X)$ coincides with $\max_{P_i}(X \mid Y)$.

\begin{remark} Each conditionally lexicographic $P_i$ induces a family $(P_i(Y)|_O)_{Y \subseteq O}$ of conditional marginal preferences. However, not every family of marginal preferences arises from a conditionally lexicographic preference: Definition~\ref{characterization of CL preferences} implies that the family $(P_i(Y)|_O)_{Y \subseteq O}$ must satisfy certain consistency constraints. For example, any two conditional marginal preferences $P_i(X)|_O$ and $P_i(Y)|_O$ must top-rank the same object. The LP-tree representation in Appendix~\ref{Appendix: LP trees} spells out these constraints precisely.
\end{remark}

An allocation $\mu$ satisfies the \defterm{worst-endowment lower bound} at a preference profile $P \in \mathcal{CL}$ if, for each $i \in N$, and each $o \in \mu_i$, $o \mathrel{R_i(\mu_i)} \min_{P_i(\mu_i)}(\omega_i)$.
In other words, \emph{conditional on receiving $\mu_i$}, every object assigned to agent~$i$ is at least as good, according to $P_i(\mu_i)|_O$, as her worst endowed object.\footnote{\WELB and \IR are logically independent on $\mathcal{CL}$. If one adds the auxiliary assumption that, for each agent, any bundle that fails \WELB is ``unacceptable'' (i.e., strictly worse than her endowment), then \IR implies \WELB on the resulting domain (see footnote~\ref{footnote other domain}). We do not impose this additional restriction.}

\begin{definition}
A rule $\varphi$ satisfies the \defprop{worst-endowment lower bound (\WELB)} if, for each $P \in \mathcal{CL}$, $\varphi(P)$ satisfies the worst-endowment lower bound at $P$.
\end{definition}

Drop strategies extend naturally to conditionally lexicographic preferences. An agent implements a drop strategy by dropping an object she does not own to the bottom of her conditionally lexicographic preferences. Formally, given $P_i \in \mathcal{CL}_i$, we say that $P'_i \in \mathcal{CL}_i$ is a \defterm{drop strategy} for $P_i$ if there exists $x \in O \setminus \omega_i$ such that
\begin{itemize}
\item[(i)] for each nonempty $Y \subseteq O \setminus \{x\}$, $Y \mathrel{P'_i} x$, and 
\item[(ii)] for all $Y,Z \subseteq O \setminus \{x\}$, $Y \mathrel{P'_i} Z$ if and only if $Y \mathrel{P_i} Z$.
\end{itemize}
In this case, we say that $P'_i$ is obtained from $P_i$ by dropping object $x$.\footnote{As in the case of purely lexicographic preferences (but not responsive preferences), there is exactly one $P'_i$ obtained from $P_i$ by dropping $x$.} Note that this definition is equivalent to the one in Section~\ref{sec:Properties} when $P_i$ is purely lexicographic. Let $\mathcal{D}_i(P_i)$ denote the set of all drop strategies for $P_i$.

Intuitively, if $P'_i$ is obtained from $P_i$ by dropping object $x$, then, for each $Y \subseteq O$, the conditional marginal preference $P'_i(Y)|_O$ is obtained from $P_i(Y \setminus \{x\})|_O$ by dropping object $x$. In particular, $x$ is the worst object in every conditional marginal preference associated with $P'_i$. The LP-tree representation in Appendix~\ref{Appendix: LP trees} shows that this family of modified conditional marginal preferences is induced by a unique conditionally lexicographic preference $P'_i \in \mathcal{CL}_i$.

\begin{definition}
A rule $\varphi$ satisfies \defprop{drop strategy-proofness (\DSP)} if, for each $P \in \mathcal{CL}$, each $i \in N$, and each $P'_i \in \mathcal{D}_i(P_i)$, we have $\varphi_i(P) \mathrel{R_i} \varphi_i(P'_i, P_{-i})$.
\end{definition}

\subsection{Augmented Top Trading Cycles}

The \defterm{Augmented Top Trading Cycles rule (ATTC)}, denoted $\varphi^{\mathrm{ATTC}}$, is the natural extension of $\varphi^{\mathrm{TTC}}$ from lexicographic to conditionally lexicographic preferences (\citet{fujita2018}). At each step of the ATTC algorithm, each agent points to her most-preferred unassigned object \emph{conditional on the objects already assigned to her}, and every unassigned object points to its owner. The resulting directed graph contains at least one trading cycle, and each agent involved in a trading cycle is assigned the object to which she points. All objects involved in a trading cycle are then removed. If unassigned objects remain, then the procedure continues to the next step; otherwise, it terminates with the resulting allocation.

Formally, given a preference profile $P \in \mathcal{CL}$, $\varphi^{\text{ATTC}}$ returns the allocation $\varphi^{\text{ATTC}}(P)$ determined by the following \defterm{ATTC algorithm} at $P$. We denote this specific instance as $\mathrm{ATTC}(P)$.

\noindent \rule{1\columnwidth}{1pt}
\paragraph{Algorithm:} $\text{ATTC}(P)$
\begin{description}
    \item [\textnormal{\emph{Initialization}:}] Set $\mu^{0}\coloneqq\left(\emptyset\right)_{i\in N}$
    and $O^{1}\coloneqq O$.
    \item [\textnormal{\emph{Step~$t \geq 1$}:}] ~
        Each agent~$i$ points to $\max_{P_i}(O^t \mid \mu_i^{t-1})$, and each object in $O^t$ points to its owner. Let $\mathcal{C}_t(P)$ be the resulting set of trading cycles. For each $C \in \mathcal{C}_t(P)$, assign to each $i \in N(C)$ the object she points to in $C$; this yields the partial allocation $\mu^t$ with
        $$
        \mu^t_i = 
        \begin{cases}
            \mu^{t-1}_i \cup \{\max_{P_i}(O^t \mid \mu_i^{t-1})\}, & \text{if } i \in \bigcup_{C \in \mathcal{C}_t(P)}N(C), \\
            \mu^{t-1}_i, & \text{otherwise.}
        \end{cases}
        $$
        Remove all objects involved in trading cycles to form the set $O^{t+1} = O^t \setminus \bigcup_{C \in \mathcal{C}_t(P)} O(C)$. Proceed to step~$t+1$ if $O^{t+1} \neq \emptyset$; proceed to Termination otherwise.
        
    \item [\textnormal{\emph{Termination}:}] Because at least one object is removed at each step, the algorithm terminates at some step $T$. Return the allocation $\varphi^{\text{ATTC}}\left(P\right)\coloneqq\mu^{T}$.
\end{description}
\rule{1\columnwidth}{1pt}

\subsection{Characterization of ATTC}

Although conditionally lexicographic preferences allow for complementarities, the domain still behaves like the lexicographic domain from the perspective of efficiency. In particular, ruling out Pareto-improving single-object exchanges is enough to guarantee full Pareto efficiency.

\begin{proposition}\label{efficiency equivalence on CL domain}
    On the conditionally lexicographic domain, an allocation satisfies \IGE if and only if it satisfies \PE.
\end{proposition}

The equivalence between \IGE and \PE implies that any myopic procedure which, at each step, greedily executes a Pareto-improving trading cycle whenever one exists must terminate at a \PE allocation on this domain. The ATTC algorithm is one such procedure.

Beyond efficiency, \citet{fujita2018} show that $\varphi^{\mathrm{ATTC}}$ is core selecting (and hence satisfies both \PE and \IR). They also prove that it is NP-hard for an agent to find a profitable manipulation, and that any successful manipulation yields only limited gains: no manipulator can obtain an object better than her most-preferred object under truthful reporting. Complementing these results, we establish that $\varphi^{\mathrm{ATTC}}$ satisfies \DSP on the conditionally lexicographic domain.

\begin{proposition}\label{ATTC is drop SP}
    On the conditionally lexicographic domain, $\varphi^{\text{ATTC}}$ satisfies \DSP.
\end{proposition}

Proposition~\ref{ATTC is drop SP} extends a result of \citet{altuntacs2023}, which shows that $\varphi^{\mathrm{TTC}}$ satisfies a strong drop strategy-proofness requirement on the lexicographic domain---one that also protects against manipulations where agents drop objects they own. We work with the weaker \DSP (that considers only drops outside the endowment), which suffices for our characterization result. Moreover, $\varphi^{\mathrm{ATTC}}$ fails the stronger variant on the conditionally lexicographic domain.

\begin{theorem}\label{characterization of ATTC}
On the conditionally lexicographic domain, a rule satisfies \BAL, \IGE, \WELB, and \DSP if and only if it equals $\varphi^{\text{ATTC}}$.
\end{theorem}

In light of Proposition~\ref{efficiency equivalence on CL domain}, Theorem~\ref{characterization of ATTC} says that $\varphi^{\mathrm{ATTC}}$ is the unique rule satisfying \BAL, \PE, \WELB, and \DSP. Thus, the theorem effectively extends Theorem~\ref{Multi-unit characterization with drop SP} from the lexicographic domain to the broader conditionally lexicographic domain, and it addresses open questions posed by \citet[p.~167]{altuntacs2023} and \citet[p.~531]{fujita2018}.

Table~\ref{TTC table} collects the key properties that $\varphi^{\mathrm{TTC}}$ and $\varphi^{\mathrm{ATTC}}$ enjoy on the three domains we consider.

\begin{table}[ht]
\caption{Properties of TTC and ATTC on various domains}\label{TTC table}
\centering
\renewcommand{\arraystretch}{1.15}
\begin{tabular}{l*{7}{c}}
\hline
Domain (Rule) & \BAL & \WELB & \IR & \PE & \IGE & \TP & \DSP \\
\hline
$\mathcal{L}$ (TTC)   & \checkmark & \checkmark & \checkmark & \checkmark & \checkmark & \checkmark & \checkmark \\
$\mathcal{R}$ (TTC)   & \checkmark & \checkmark & \checkmark &            & \checkmark & \checkmark &            \\
$\mathcal{CL}$ (ATTC) & \checkmark & \checkmark & \checkmark & \checkmark & \checkmark & ---        & \checkmark \\
\hline
\end{tabular}

\vspace{0.5em}
\begin{minipage}{0.9\textwidth}
    \footnotesize
    \textbf{Notes:} A \checkmark\ indicates that the rule satisfies the corresponding property on the given domain. 
    On $\mathcal{L}$ and $\mathcal{CL}$, \PE and \IGE coincide. We do not study \TP on $\mathcal{CL}$.
    \hfill $\diamond$
\end{minipage}
\end{table}

We conclude this section by showing the difficulty in obtaining positive results in more general preference domains. On any domain of monotonic preferences that strictly includes the conditionally lexicographic domain, the equivalence between \PE and \IGE breaks down. Thus, on larger domains, one must look beyond single-object exchanges in order to achieve efficiency.\footnote{If we broaden the definition of allocations to allow agents to be assigned empty bundles, then the conditionally lexicographic domain is a maximal domain \emph{among all strict preferences} on which \PE and \IGE coincide.}

\begin{proposition}\label{maximal domain 1}
Within the domain of monotonic preferences, the conditionally lexicographic domain is a maximal domain on which \IGE and \PE are equivalent.
\newline
(More precisely, if $\mathcal{CL} \subsetneq \mathcal{P} \subseteq \mathcal{M}$, then there is a set $N$ of agents, a preference profile $P \in \mathcal{P}$, and an allocation $\mu \in \mathcal{A}$ such that $\mu$ satisfies \IGE but not \PE at $P$.)
\end{proposition}

\section{Related literature and discussion}\label{related literature}

Since its introduction for housing markets in the seminal paper by \citet{shapley1974}, TTC has become the subject of a large body of literature; see \citet{morrill2024top} for a survey. For housing markets with strict preferences, TTC is the canonical rule: it selects the unique core allocation (\citet{shapley1974,roth1977}), it is group strategy-proof (\citet{roth1982,bird1984}), and it is the unique rule satisfying \PE, \IR, and \SP (\citet{ma1994}). Theorem~\ref{Theorem: Single-unit characterization} strengthens \citeauthor{ma1994}'s (\citeyear{ma1994}) characterization by replacing \SP with the weaker incentive requirement \TP, leaving little scope for alternative rules in housing markets; see Section~\ref{sec:SS} for further discussion.

Early work on multi-object reallocation consists mostly of negative results. On general preference domains, \PE, \IR, and \SP are incompatible once some agent owns more than one object (\citet{sonmez1999}), and the impossibility persists even under lexicographic preferences (\citet{todo2014}). Moreover, the core can be empty even under additive (thus responsive) preferences (\citet{konishi2001}).

Recent studies obtain positive findings under restricted preference domains. For instance,  \PE, \IR, and \SP are compatible when agents' bundle preferences are based on ``dichotomous'' (\citet{andersson2021organizing}) or  ``trichotomous'' (\citet{manjunath2021,manjunath2024}) marginal preferences over objects (see also \citet{han2024blood}). In addition, under conditionally lexicographic preferences the core is guaranteed to be nonempty, and ATTC selects an allocation in the core (\citet{fujita2018}). Since the core need not be single-valued, our characterizations under (conditionally) lexicographic preferences provide a normative justification for core selection via (augmented) TTC (see Theorems~\ref{Multi-unit characterization with TP}, \ref{Multi-unit characterization with drop SP}, and~\ref{characterization of ATTC}).

In the closest paper to ours, \citet{altuntacs2023} consider the same model but focus on the lexicographic domain. They show that TTC is characterized by \PE, \DSP, and a property termed the \defterm{strong endowment lower bound}, which requires each agent's assignment to ``pairwise dominate'' her endowment according to her marginal preferences.\footnote{Formally, an allocation $\mu$ satisfies the \defterm{strong endowment lower bound} at $P$ if, for each agent $i$, there is a bijection $f_i : \omega_i \to \mu_i$ such that, for each $ o \in \omega_i$, $f_i(o) \mathrel{R_i} o$.} The strong endowment lower bound implies \BAL and \WELB on any preference domain, and under responsive preferences it also implies \IR (and \citeauthor{papai2003}'s (\citeyear{papai2003}) ``strong individual rationality''). Consequently, we establish the same uniqueness under substantially weaker criteria---replacing their strong endowment lower bound with \BAL and \WELB (Theorems~\ref{Multi-unit characterization with TP} and~\ref{Multi-unit characterization with drop SP}), and substituting \DSP with \TP (Theorem~\ref{Multi-unit characterization with TP}). We also obtain characterizations on richer domains---namely, the responsive domain (where TTC fails \PE and \DSP but satisfies the meaningful compromises \IGE and \TP) and the conditionally lexicographic domain (where ATTC satisfies \PE, \IR, and \DSP)---addressing open questions posed by \citet[p.~167]{altuntacs2023}.

We prove the uniqueness of TTC with a minimal-counterexample argument, using a novel minimality condition that combines the size criterion of \citet{sethuraman2016} and the similarity criterion of \citet{ekici2022}. Our proofs of Theorems~\ref{Multi-unit characterization with TP} and~\ref{characterization of ATTC} (and the omitted proof of Theorem~\ref{Theorem: Single-unit characterization}) are minimal in the sense that they invoke each axiom exactly once, further illustrating how the general approach of \citet{sethuraman2016} and \citet{ekici2022} can be quite fruitful when paired with a well-chosen minimality criterion. Our proof of Proposition~\ref{ATTC is drop SP} is based on a modified ATTC algorithm that is outcome-equivalent to the original ATTC algorithm; a similar modification is used by \citet{gonczarowski2023} to ``expose'' the strategy-proofness of TTC in single-object environments.

Related studies consider adaptations of TTC to multi-object reallocation problems with  additional structure. \citet{biro2022,biro2022serial} consider a multiple-copy Shapley-Scarf model, where each agent owns multiple copies of a homogeneous, agent-specific object (equivalently, every agent is indifferent between all objects with the same initial owner). Focusing on marginal rules for responsive preferences, \citet{biro2022} characterize a variant of TTC using subset-drop strategy-proofness, and they determine the capacity configurations under which this variant satisfies \PE, \IR, and \SP. Similarly, \citet{feng2022} characterize another variant of TTC for \citeauthor{moulin1995}'s (\citeyear{moulin1995}) multiple-type housing markets, where objects are partitioned into types and each agent consumes exactly one object of each type.


Our characterizations of TTC under responsive preferences complement characterizations of strategy-proof alternatives on the responsive domain. Strategy-proofness sharply restricts the admissible rules, leading to rather extreme rules under mild auxiliary axioms. \defterm{Sequential Dictatorships} are characterized by \SP, \PE, and ``non-bossiness''  (e.g., \citet{papai2001,ehlers2003scw,hatfield2009strategy,monte2015}), whereas \defterm{Segmented Trading Cycles rules} are characterized by \SP, strong individual rationality, non-bossiness, and ``trade sovereignty'' (\citet{papai2003,papai2007}). Our characterizations identify TTC as the sensible middle ground: though it fails two of the three ideals on the responsive domain (\PE and \SP), it delivers meaningful guarantees along all three dimensions (\IGE, \IR/\WELB, and \TP).

We relax \SP by ruling out restricted classes of simple manipulations: subset-drop, drop, and truncation strategies. This focus is consistent with computational results emphasizing the difficulty of manipulating TTC in practice (e.g., \citet{fujita2018,phan2022}). Drop strategies and their variants are studied by \citet{altuntacs2023} for multi-object reallocation under lexicographic preferences. 
Truncation has roots in the literature on matching markets, where it corresponds to shortening a list of acceptable partners (\citet{mongell1991sorority,roth1991incentives,roth1999truncation,ehlers2008}). Truncation strategies also appear in matching with contracts (e.g., \citet{afacan2016characterizations,hatfield2021stability}) and school-choice problems (\citet{shirakawa2024simple}), though our notion is closest to properties studied by \citet{kojima2013efficient} for multi-object assignment and \citet{biro2022,biro2022serial} for multi-object reallocation.

Finally, when relaxing marginality and considering more complex reporting languages, competitive-equilibrium approaches become natural benchmarks. Existence issues are typically avoided either by seeking equilibria in pseudomarkets for probability shares (e.g., \citet{hylland1979,echenique2021fairness,echenique2023balanced}) or by accepting only approximate market clearing (e.g., \citet{budish2011combinatorial,nguyen2022near,jantschgi2025competitive}). These procedures offer compelling efficiency and fairness properties on general preference domains, but they typically suffer from onerous reporting requirements and considerable computational complexity. Our study of ATTC under conditionally lexicographic preferences provides a practical alternative: its reporting language comprises LP-trees which are expressive but compact, and it delivers strong welfare guarantees using a procedure that is computationally tractable.

\newpage
\appendix

\section{Lexicographic preference trees}\label{Appendix: LP trees}

This appendix introduces lexicographic preference trees, which provide a compact representation of conditionally lexicographic preferences.

\begin{definition}
\label{LP definition}
A \defterm{lexicographic preference tree (LP tree)} on a nonempty subset $X \subseteq O$ is a directed binary tree $\tau_{i}$ such that
\begin{enumerate}
    \item each vertex $v$ is labeled with an object $o(v) \in X$.
    \item every object in $X$ appears exactly once on any path from the root to a leaf.
    \item every internal (non-leaf) vertex $v$ has two outgoing edges: (i) an ``in edge'' labeled $o(v)$, and (ii) a ``not-in edge'' labeled $\neg o(v)$.
\end{enumerate}
An LP tree on $O$ is referred to more succinctly as an \defterm{LP tree}.
\end{definition}

Given an LP tree $\tau_i$ and a bundle $X\subseteq O$, let $\tau_{i}\left(X\right) = (v_1, v_2, \dots, v_{|O|})$ be the unique root-to-leaf path in $\tau_i$ containing only edges consistent with $X$: for each $k = 1,\dots,|O|-1$, traverse the ``in edge'' from $v_k$ if $o(v_k) \in X$, and the ``not-in edge'' from $v_k$ if $o(v_k) \notin X$.
Given bundles
$A$ and $B$ with $A \neq B$, let $\tau_i(A,B)$ denote the last vertex common to both $\tau_{i}\left(A\right)$ and $\tau_{i}\left(B\right)$. Equivalently,  $\tau_{i}\left(A,B\right)$ is
the first vertex $v$ on $\tau_{i}\left(A\right)$ and
$\tau_{i}\left(B\right)$ with $o\left(v\right)\in (A \setminus B) \cup (B \setminus A)$. Figures~\ref{fig:LP_tree1} and \ref{fig:LP_tree_lexicographic1} provide a graphical representation of two LP trees on $O = \{a, b, c, d\}$, together with the paths consistent with $\{a,b \}$ and $\{a, c\}$.

\begin{definition}
The preference relation $P_{\tau_{i}}$ associated with an LP tree
$\tau_{i}$ on $O$ is defined by: 
\begin{enumerate}
\item for all $A,B\subseteq O$ with $A\neq B$, $\left[A \mathrel{P_{\tau_i}} B \iff o\left(\tau_i\left(A,B\right)\right)\in A\setminus B\right]$.
\item for all $A,B\subseteq O$, $\left[A \mathrel{R_{\tau_i}} B \iff\left(A=B\text{ or }A \mathrel{P_{\tau_i}}B\right)\right]$.
\end{enumerate}
\end{definition}

\begin{figure}[htb]
\centering
\begin{minipage}{0.46\textwidth}
    \centering
    \includegraphics[width=\textwidth]{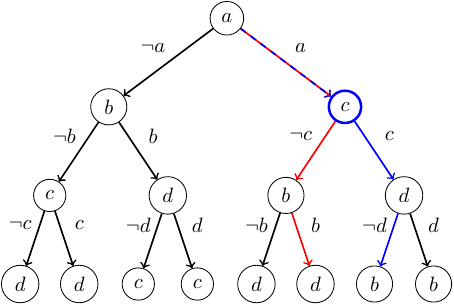}
    \caption{An LP tree $\tau_i$}
    \label{fig:LP_tree1}
\end{minipage}\hfill
\begin{minipage}{0.46\textwidth}
    \centering
    \includegraphics[width=\textwidth]{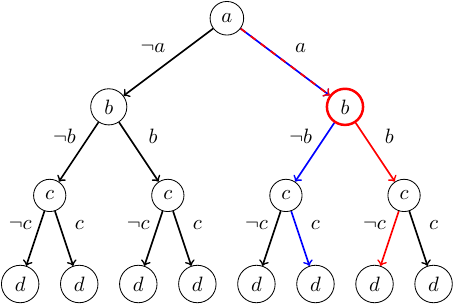}
    \caption{An LP tree $\tau^*_i$}
    \label{fig:LP_tree_lexicographic1}
\end{minipage}
\vspace{1em} 

\begin{minipage}{\textwidth}
    \small \textbf{Notes:} In both figures, the paths corresponding to the bundles $\{a,b\}$ and $\{a,c\}$ are highlighted in \textcolor{red}{red} and \textcolor{blue}{blue}, respectively. In Figure~\ref{fig:LP_tree1}, the last common vertex, $\tau_i(\{a, b\}, \{a, c\})$, is highlighted in \textcolor{blue}{blue}. Because $o(\tau_i(\{a, b\}, \{a, c\})) = c$ belongs to $\{a,c\} \setminus \{a,b\}$, we have $\{a,c\} \mathrel{P_{\tau_i}} \{a,b\}$. In Figure~\ref{fig:LP_tree_lexicographic1}, the last common vertex, $\tau^*_i(\{a, b\}, \{a, c\})$, is highlighted in \textcolor{red}{red}. Because $o(\tau^*_i(\{a, b\}, \{a, c\})) = b$ belongs to $\{a,b\} \setminus \{a,c\}$, we have $\{a,b\} \mathrel{P_{\tau^*_i}} \{a, c\}$. Since all paths in $\tau^*_i$ use the same object order, $P_{\tau^*_i}$ is purely lexicographic. \hfill $\diamond$
    
\end{minipage}
\end{figure}

A preference relation $P_{i}$ on $2^O$ is conditionally
lexicographic if and only if there exists an LP tree $\tau_{i}$ (on $O$) such that
$P_{i}=P_{\tau_{i}}$. In fact, there is a one-to-one correspondence between LP trees and conditionally lexicographic preferences on $O$, and we denote the unique LP tree associated with $P_i$ by $\tau_{P_{i}}$.
A lexicographic preference corresponds to an LP tree where all root-to-leaf paths use the same object order (i.e., all vertices at the same depth are labeled with the same object).

\paragraph{Worst-endowment lower bound.}
Given an LP tree $\tau_i$ and a vertex $v$ of $\tau_{P_i}$,
let $a\left(v\right)$ denote the set of objects labeling the ancestors
of $v$, including $v$ itself. That is, if $\left(v_{1},v_{2},\dots,v_{k}=v\right)$ is the
unique path from the root of $\tau_i$ to $v$, then $a\left(v\right)=\left\{ o\left(v_{1}\right),o\left(v_{2}\right),\dots,o\left(v_{k}\right)\right\} $.
Given a preference relation $P_i \in \mathcal{CL}_i$ and a bundle $X\in2^{O}$,
define $w_{P_i}(\omega_i \mid X)$ as the last vertex on the path $\tau_{P_i}(X)$ whose label belongs to $\omega_i$. Equivalently, $w_{P_{i}}\left(\omega_{i}\mid X\right)$ is the unique vertex
on $\tau_{P_{i}}\left(X\right)$ such that $o\left(w_{P_{i}}\left(\omega_{i}\mid X\right)\right)\in\omega_{i}$
and $\omega_{i}\subseteq a\left(w_{P_{i}}\left(\omega_{i}\mid X\right)\right)$. We illustrate this construction graphically in figures~\ref{fig:LP_tree1_WELB} and \ref{fig:LP_tree_lexicographic1_WELB}.

An allocation $\mu$ satisfies the \defprop{worst-endowment lower bound}
at a preference profile $P \in \mathcal{CL}$ if, for each $i\in N$, $\mu_{i}\subseteq a\left(w_{P_{i}}\left(\omega_{i}\mid\mu_{i}\right)\right)$.
This means every object in $\mu_{i}$ labels a vertex on $\tau_{P_{i}}\left(\mu_{i}\right)$
that appears before (or equals) vertex $w_{P_{i}}\left(\omega_{i}\mid\mu_{i}\right)$.
Intuitively, conditional on receiving $\mu_i$, no object assigned to agent~$i$ is worse than every object in her endowment $\omega_{i}$. When agents have purely lexicographic preferences, this definition aligns with the original \WELB.

\begin{figure}[htb]
\centering
\begin{minipage}{0.46\textwidth}
    \centering
    \includegraphics[width=\textwidth]{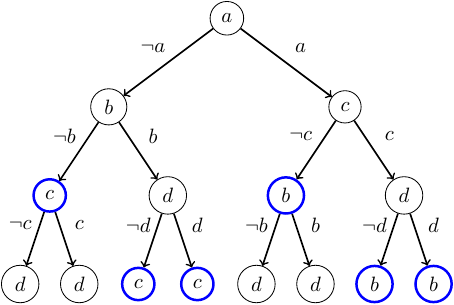}
    \caption{The LP tree $\tau_i$}
    \label{fig:LP_tree1_WELB}
\end{minipage}\hfill
\begin{minipage}{0.46\textwidth}
    \centering
    \includegraphics[width=\textwidth]{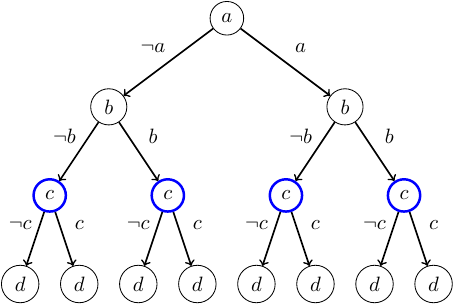}
    \caption{The LP tree $\tau^*_i$}
    \label{fig:LP_tree_lexicographic1_WELB}
\end{minipage}
\vspace{1em} 

\begin{minipage}{\textwidth}
    \small \textbf{Notes:} Figures~\ref{fig:LP_tree1_WELB} and \ref{fig:LP_tree_lexicographic1_WELB} illustrate the construction of $w_{P_{i}}(\omega_{i} \mid X)$. Suppose agent~$i$'s endowment is $\omega_i = \{b, c\}$. In both figures, the vertices corresponding to $w_{P_i}(\omega_i \mid X)$ for various $X  \in 2^O$ are highlighted in \textcolor{blue}{blue}. Notice that on any path from the root to a leaf, there is exactly one such vertex. Under $\tau_i$, \WELB precludes assigning agent~$i$ object $d$ together with any of the bundles $\emptyset, \{a\}, \{a,b\}$, or $\{c\}$; that is, she does not receive a bundle in $\{\{d\}, \{a,d\}, \{a,b,d\}, \{c,d\}\}$. Under $\tau^*_i$, \WELB precludes any bundle containing object $d$. \hfill $\diamond$ 
\end{minipage}
\end{figure}
\paragraph{Drop strategies.}

Given $P_i \in \mathcal{CL}_i$, let $P'_i$ be obtained from $P_i$ by dropping object $x$. Then $P'_i$ is represented by the LP tree $\tau_{P'_i}$, obtained from $\tau_{P_i}$ by moving $x$ to the bottom of every root-to-leaf path. Formally, let $\tau_{P_i}^{-x}$ denote the LP tree on $O \setminus \{x\}$ representing the restriction of $P_i$ to subsets of $O \setminus \{x\}$.\footnote{We can construct $\tau_{P_i}^{-x}$ as follows. For each vertex $v$ of $\tau_{P_i}$, let $T_v$ denote the maximal subtree of $\tau_{P_i}$ consisting of a vertex $v$ of $\tau_{P_i}$ together with all of its successors, and let $v'$ be the child of $v$ whose incoming edge $(v, v')$ is labeled with $\neg o(v)$. For each vertex $v$ of $\tau_{P_i}$ such that $o(v) = x$, simply replace $T_v$ with the subtree $T_{v'}$ (or the empty tree if $v$ is a leaf of $\tau_{P_i}$).} To construct $\tau_{P'_i}$ from $\tau_{P_i}^{-x}$, at each leaf of $\tau_{P_i}^{-x}$ append two child nodes labeled with $x$, connected by an ``in edge'' and a ``not-in edge'' accordingly. This modification ensures that $x$ is the least desirable addition to any bundle. Figures~\ref{fig:LP_tree_dropStrategy_color1} and \ref{fig:LP_tree_dropStrategy_color2} illustrate this construction.


\begin{figure}[htb]
\centering
\begin{minipage}{0.46\textwidth}
    \centering
    \includegraphics[width=\textwidth]{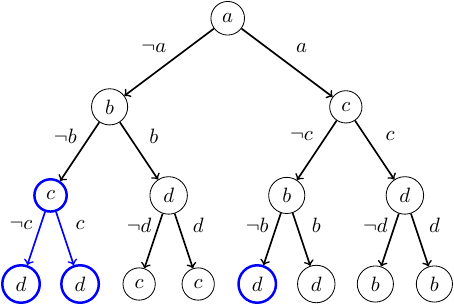}
    \caption{The LP tree $\tau_i$}
    \label{fig:LP_tree_dropStrategy_color1}
\end{minipage}\hfill
\begin{minipage}{0.46\textwidth}
    \centering
    \includegraphics[width=\textwidth]{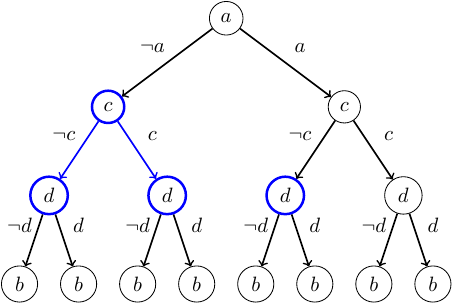}
    \caption{The LP tree $\tau'_i$ obtained by dropping $b$}
    \label{fig:LP_tree_dropStrategy_color2}
\end{minipage}
\vspace{1em} 

\begin{minipage}{\textwidth}
    \small \textbf{Notes:} Figure~\ref{fig:LP_tree_dropStrategy_color1} displays the original LP tree from Figures~\ref{fig:LP_tree1} and \ref{fig:LP_tree1_WELB}. To illustrate the drop strategy where agent~$i$ drops object $b$, we construct the LP tree $\tau'_i$ shown in Figure~\ref{fig:LP_tree_dropStrategy_color2}. At each internal vertex $v$ labeled $b$, we do the following: (i) remove the subtree consisting of $v$ and all of its descendants, and (ii) append the subtree, shown in \textcolor{blue}{blue}, consisting of the child of $v$ reached via the ``not-in edge'' together with all of its descendants. Then, at each leaf of this modified tree, we append two child nodes labeled $b$, connected via an ``in edge'' and a ``not-in edge'' accordingly. \hfill $\diamond$
\end{minipage}
\end{figure}

\newpage

\section{Proofs}\label{appendix:proof}

\subsection{Proof of Lemma~\ref{individual-good + IR implies WELB}}
Suppose $\varphi$ is a marginal rule on $\mathcal{R}$ satisfying \IR. Let $P \in \mathcal{R}$ and denote $\mu \coloneqq \varphi(P)$.

Toward contradiction, suppose that $\mu$ violates \BAL at $P$. Then there exists an agent $i$ with $|\mu_i| > |\omega_i|$. Because responsiveness imposes no restrictions on comparisons among bundles with different cardinalities (and it does not require monotonicity), there exists $P^*_i \in \mathcal{R}_i$ with $P^*_i|_O = P_i|_O$ such that, for each bundle $X \subseteq \omega_i$,  $X \mathrel{P^*_i} \mu_i$.\footnote{\label{footnote: responsive preference construction}For instance, let $P^*_i$ order bundles first by cardinality (with smaller bundles preferred to larger bundles), and then order bundles with the same cardinality using the lexicographic order associated with $P_i|_O$. For example, if $O = \{a, b, c\}$ and $P_i|_O : a, b, c$, then $P^*_i : \emptyset , \{a\}, \{b\}, \{c\}, \{a, b\}, \{a, c\}, \{b, c\}, \{a, b, c\}$.} Then marginality of $\varphi$ implies that $\varphi(P^*_i,P_{-i}) = \varphi(P) = \mu$. By construction of $P^*_i$, we have $\omega_i \mathrel{P^*_i} \mu_i = \varphi_i(P^*_i, P_{-i})$, which contradicts \IR. Thus $\varphi$ satisfies \BAL.

Now suppose for a contradiction that $\mu$ violates \WELB at $P$. Then there is some agent $i \in N$ such that $\min_{P_i}(\omega_i) \mathrel{P_i} \min_{P_i}(\mu_i)$.
By \BAL, we also have $|\mu_i| = |\omega_i|$. Consequently, there exists $P'_i \in \mathcal{R}_i$ with $P'_i |_O = P_i |_O$ and $\omega_i \mathrel{P'_i} \mu_i$. However, marginality of $\varphi$ implies that $\varphi_i(P'_i, P_{-i}) = \varphi_i(P) = \mu_i$. This contradicts \IR.
\hfill{}$\blacksquare$

\subsection{Proof of Lemma~\ref{Proposition: drop SP + WELB implies subset-drop SP} }

Suppose $\varphi$ satisfies \DSP and \WELB. Let $P'_{i}$ be obtained from
$P_{i}$ by dropping some subset $X \subseteq O \setminus \omega_{i}$.
Suppose that $X=\left\{ x_{1},x_{2},\dots,x_{k}\right\} $, where
$x_{1} \mathrel{P_{i}} x_{2} \mathrel{P_{i}} \cdots \mathrel{P_{i}} x_{k}$. Then $P'_{i}$ is obtained
from $P_i$ via a sequence of $k$ drop strategies. That is, $P'_{i}=P_{i}^{k}$,
where $P_{i}^{0}=P_{i}$ and $P_{i}^{1},\dots,P_{i}^{k}$ are such
that, for each $\ell\in\left\{ 1,\dots,k\right\} $, $P_{i}^{\ell}$
is obtained from $P_{i}^{\ell-1}$ by dropping object $x_{\ell}$. 
\begin{claim}
\label{Claim 1}For each $\ell =  1,\dots,k$, $\varphi_{i}(P_{i}^{\ell-1},P_{-i}) \mathrel{R_{i}} \varphi_{i}(P_{i}^{\ell},P_{-i})$.
\end{claim}
\begin{proof} 
The proof is by induction on $\ell$. Clearly, $\varphi_{i}(P)=\varphi_{i}(P_{i}^{0},P_{-i}) \mathrel{R_{i}} \varphi_{i}(P_{i}^{1},P_{-i})$
by \DSP. For the inductive step, suppose
that $\ell \in \{ 1,\dots,k-1\} $ is such that 
\[
\varphi_{i}(P)=\varphi_{i}(P_{i}^{0},P_{-i}) \mathrel{R_{i}} \varphi_{i}(P_{i}^{1},P_{-i}) \mathrel{R_{i}} \cdots \mathrel{R_{i}} \varphi_{i}(P_{i}^{\ell},P_{-i}).
\]
It suffices to show that $\varphi_{i}(P_{i}^{\ell},P_{-i}) \mathrel{R_{i}} \varphi_{i}(P_{i}^{\ell+1},P_{-i})$.
By \DSP, we have $\varphi_{i}(P_{i}^{\ell},P_{-i}) \mathrel{R_i^{\ell}} \varphi_{i}(P_{i}^{\ell+1},P_{-i})$. Moreover, \WELB implies that $$\varphi_{i}(P_{i}^{\ell},P_{-i}) \subseteq O\setminus \{ x_{1},\dots,x_{\ell}\} 
\quad \text{and} \quad
\varphi_{i}(P_{i}^{\ell+1},P_{-i}) \subseteq O \setminus \{ x_{1},\dots,x_{\ell+1}\} \subseteq O \setminus \{ x_{1},\dots,x_{\ell}\}.$$
Because $P_{i}^{\ell} |_{O \setminus \{ x_{1},\dots,x_{\ell}\}} = P_{i} |_{O \setminus \{ x_{1},\dots,x_{\ell}\}}$, and $P_i$ and $P^{\ell}_i$ are lexicographic,
we have $\varphi_{i}(P_{i}^{\ell},P_{-i}) \mathrel{R_{i}} \varphi_{i}(P_{i}^{\ell+1},P_{-i})$,
as desired.
\end{proof}
It follows from Claim \ref{Claim 1} that 
\[
\varphi_{i}(P)=\varphi_{i}(P_{i}^{0},P_{-i}) \mathrel{R_{i}} \varphi_{i}(P_{i}^{1},P_{-i}) \mathrel{R_{i}} \cdots \mathrel{R_{i}} \varphi_{i}(P_{i}^{k},P_{-i}) = \varphi_i(P'_i, P_{-i}).
\]
Therefore, $\varphi_{i}(P) \mathrel{R_{i}} \varphi_i(P'_i, P_{-i})$. \hfill{}$\blacksquare$

\subsection{Proof of Lemma~\ref{lemma:BAL implied on responsive domain}}
Suppose $\varphi$ is a marginal rule on $\mathcal{R}$ satisfying \WELB and \TP. Let $P \in \mathcal{R}$ and denote $\mu \coloneqq \varphi(P)$.

Toward contradiction, suppose that $\mu \coloneqq \varphi(P)$ violates \BAL at $P$. Then there is an agent $i \in N$ such that $|\mu_i| > |\omega_i|$. Because responsiveness imposes no restrictions on comparisons among bundles with different cardinalities (and it does not require monotonicity), there exists $P^*_i \in \mathcal{R}_i$ such that $P^*_i|_O = P_i|_O$ and, for each bundle $X \subseteq \omega_i$, $X \mathrel{P^*_i} \mu_i$ (see footnote~\ref{footnote: responsive preference construction}). Then marginality of $\varphi$ implies that $\varphi(P^*_i, P_{-i}) = \varphi(P) = \mu$.

Now let $P'_i \in \mathcal{R}_i$ be a truncation of $P^*_i$ obtained by dropping the set $O \setminus \omega_i$ (thus $P'_i|_O$ ranks every object in $\omega_i$ above all others). 
Then \WELB implies that $\varphi_i(P'_i, P_{-i}) \subseteq \omega_i$. By construction of $P^*_i$, we have $\varphi_i(P'_i, P_{-i}) \mathrel{P^*_i} \mu = \varphi_i(P^*_i,P_{-i})$. Thus $\varphi$ fails \TP, a contradiction.
\hfill{}$\blacksquare$

\subsection{Proof of Theorem~\ref{Multi-unit characterization with TP}}

Toward contradiction, suppose that $\varphi$ satisfies the properties but $\varphi \neq \varphi^{\text{TTC}}$. We will select a preference profile according to a particular minimality criterion, and then show that the minimality of the profile leads to a contradiction. We start by introducing some notation.

For each $P \in \mathcal{L}$ and each $t \in \mathbb{N}$, recall that $\mathcal{C}_t(P)$ denotes the set of trading cycles that arise at step $t$ of $\text{TTC}(P)$. (We assume that $\mathcal{C}_t(P) = \emptyset$ if $\text{TTC}(P)$ terminates before step $t$.) Define the ``size'' function $s:\mathcal{L} \to \mathbb{N}$ such that, for each $P \in \mathcal{L}$, 
$$
s(P) = \sum_{i\in N}\left|\left\{ o\in O\mid o \mathrel{R_{i}} \min_{P_{i}}\left(\omega_{i}\right)\right\} \right|.
$$
The ``similarity'' function $\rho: \mathcal{L} \to \mathbb{N} \cup \{\infty\}$ measures the earliest point of divergence between $\varphi$ and $\text{TTC}(P)$. Formally, for each $P \in \mathcal{L}$, 
$$
\rho(P) = \inf \{t \in \mathbb{N} \mid \varphi(P) \text{ does not execute each trading cycle in } \mathcal{C}_t(P)\},$$
where it is understood that $\rho(P) = \infty$ precisely when $\varphi(P) = \varphi^{\text{TTC}}(P)$.

Among all profiles that minimize the similarity function $\rho$, let $P$ be one that further minimizes the size function $s$. Thus, for any profile $P' \in \mathcal{L}$, either (i) $\rho(P)<\rho(P')$, or (ii) $\rho(P)=\rho(P')$ and $s(P) \leq s(P')$.

Let $t \coloneqq \rho(P)$. Then $\varphi(P)$ executes each trading cycle in $\bigcup_{\tau = 1}^{t-1} \mathcal{C}_\tau(P)$, but it does not execute some trading cycle $C \in \mathcal{C}_t(P)$. Suppose that
\[
C=\left(o_1, i_{1},o_2,i_{2},\dots,o_k,i_k,o_{k+1}=o_1\right).
\]
Because $\varphi(P)$ does not execute $C$, there is an agent $i_\ell \in N(C)$ who is not assigned the object she points to within $C$, i.e., $o_{\ell+1} \notin \varphi_{i_\ell}(P)$. Without loss of generality, let $i_\ell = i_k$. Thus, $o_{k+1} = o_1 \notin \varphi_{i_k}(P)$. 

Let $O^t$ denote the set of objects available at the beginning of Step $t$ of $\text{TTC}(P)$, i.e., $O^t \coloneqq O \setminus \left(\bigcup_{\tau = 1}^{t-1} \bigcup_{C' \in \mathcal{C}_\tau(P)} O(C')\right)$. Then $\varphi_{i_k}^{\text{TTC}}(P) \setminus O^t$ is the set of objects assigned to agent~$i_k$ before step~$t$ of $\text{TTC}(P)$; moreover $\varphi_{i_k}(P) \setminus O^t = \varphi^{\text{TTC}}_{i_k}(P) \setminus O^t$.

\begin{claim}
$\varphi_{i_k}(P) \cap O^t = \omega_{i_k} \cap O^t$. \label{claim2}
\end{claim}
\begin{proof}  
Suppose otherwise. By \BAL and the fact that $\left|\varphi_{i_{k}}\left(P\right)\setminus O^{t}\right|=\left|\omega_{i_{k}} \setminus O^{t}\right|$,
we must have $\left|\varphi_{i_{k}}\left(P\right)\cap O^{t}\right|=\left|\omega_{i_{k}}\cap O^{t}\right|$.
Consequently, $\varphi_{i_{k}}\left(P\right)\cap O^{t}\neq\omega_{i_{k}}\cap O^{t}$
implies there is an object $o' \in  \varphi_{i_k}(P) \cap O^t$ with $o' \not \in \omega_{i_k}$.

By the definition of $\text{TTC}(P)$, we know that $o_1 = \max_{P_{i_k}}(O^t)$. Because $o' \in \varphi_{i_k}(P) \cap O^t$ and $o_1 \notin \varphi_{i_k}(P)$, it follows that $o_1 \mathrel{P_{i_k}} o'$. Furthermore, \WELB implies that $o'\mathrel{P_{i_k}} \min_{P_{i_k}}(\omega_{i_k}).$
 
 Let $P'_{i_k}$ be the truncation of $P_{i_k}$ at $o_1$. Thus, $P'_{i_k}$ is obtained from $P_{i_k}$ by dropping the set $\{o \in O \setminus \omega_{i_k} \mid o_1 \mathrel{P_{i_k}} o\}$.
 
 Let $P' \coloneqq (P'_{i_k}, P_{-i_k})$. Then $o_1 \mathrel{P_{i_k}} o' \mathrel{P_{i_k}} \min_{P_{i_k}}(\omega_{i_k})$ implies that $s(P') < s(P)$. Consequently, the choice of $P$ implies that $\rho(P') > \rho(P) = t$.

By definition of $P'$, the two runs $\text{TTC}(P)$ and $\text{TTC}(P')$ execute precisely the same trading cycles at every step $\tau \leq t$, i.e., $\mathcal{C}_\tau(P)= \mathcal{C}_\tau(P')$ for $\tau=1,\ldots,t$. 
Thus, $\rho(P') > t$ implies that $\varphi(P')$ executes all trading cycles in $\bigcup_{\tau = 1}^{t}\mathcal{C}_\tau(P)$. In particular,
$$(\varphi_{i_k}(P) \setminus O^t) \cup \{o_1\} \subseteq \varphi_{i_k}(P').$$
Since $P_{i_k}$ is lexicographic and $\max_{P_{i_k}}(O^t) = o_{1} \notin \varphi_{i_k}(P)$, we must have $\varphi_{i_k}(P') \mathrel{P_{i_k}} \varphi_{i_k}(P)$. This contradicts \TP.
\end{proof}

Next, we show that at least two agents are involved in $C$.

\begin{claim}\label{claim3}
	$|N(C)|>1$.
\end{claim}
\begin{proof}
Toward contradiction, suppose that $|N(C)|=1$. Thus  $i_1 = i_k$ and $C = (o_{1}, i_1, o_{1})$. Then $o_{1} \in \omega_{i_k}$ and Claim~\ref{claim2} implies that $o_{1} \in \varphi_{i_k}(P)$, a contradiction.
\end{proof}

By Claim~\ref{claim2} and Claim~\ref{claim3}, agent~$i_{k-1}$ points to $o_{k}$ on $C$, yet she does not receive $o_{k}$. An argument similar to the proof of Claim~\ref{claim2} implies that $\varphi_{i_{k-1}}(P) \cap O^t = \omega_{i_{k-1}} \cap O^t$. 

Proceeding by induction, we conclude that for each agent~$i_\ell$ involved in $C$, we have $\varphi_{i_\ell}(P) \cap O^t = \omega_{i_\ell} \cap O^t$. Thus, $C$ is a Pareto-improving trading cycle at $P$. It follows that $\varphi$ is not \IGE, a contradiction.
\hfill{}$\blacksquare$

\subsection{Proof of Proposition~\ref{prop:possibility}}

Fix $P \in \mathcal{R}$ and $i \in N$, and let $P'_i$ be a truncation of $P_i$ at an object $x \in O$.\footnote{If $P'_i$ is obtained from $P_i$ by dropping $O \setminus \omega_i$, then trivially $\varphi_i^{TTC}(P) \mathrel{R_i} \varphi_i^{TTC}(P'_i,P_{-i}) = \omega_i$.} Denote $P' \coloneqq (P'_i, P_{-i})$ and $U \coloneqq \{o \in O \mid o \mathrel{R_i} x\}$, and observe that $P_i|_U = P'_i|_U$. Consider two cases.

\begin{casenv}
    \item Suppose $\varphi_i^{\text{TTC}}(P) \subseteq U$. Then, by the definition of $\text{TTC}(P)$ and $\text{TTC}(P')$, the allocations $\varphi^\text{TTC}(P)$ and $\varphi^{TTC}(P')$ execute precisely the same trading cycles. Thus $\varphi_i^{\text{TTC}}(P) = \varphi_i^{\text{TTC}}(P')$.
    \item Suppose $\varphi_i^{\text{TTC}}(P) \nsubseteq U$. Then agent~$i$ receives at least one object in $O \setminus U$ at $\varphi^{\text{TTC}}(P)$. Let $t$ be the earliest step of $\text{TTC}(P)$ at which agent~$i$ points to an object in $O \setminus U$. 

Let $O^t$ be the set of objects remaining at the beginning of step $t$ of $\text{TTC}(P)$. Since $P_i |_U = P'_i |_U$, we know that $\text{TTC}(P)$ and $\text{TTC}(P')$ execute the same trading cycles at steps $1, \dots, t-1$. Thus, $O^t$ is the set of objects remaining at step $t$ of $\text{TTC}(P')$. Moreover, we have
\begin{equation}
    \varphi^{\text{TTC}}_i(P) \setminus O^t=\varphi^{\text{TTC}}_i(P') \setminus O^t. \label{plusone}
\end{equation}

By construction, $P'_i$ ranks every object in $\omega_i \cap O^t$ above every object in $O^t \setminus \omega_i$. 
Thus, at each step $\tau \geq t$ of $\text{TTC}(P')$ where agent~$i$ is involved in a trading cycle, she is assigned an object in $\omega_i$. Consequently,
\begin{equation}
\varphi^{\text{TTC}}_i(P') \cap O^t = \omega_i \cap O^t. \label{plustwo}
\end{equation}
On the other hand, at each step $\tau \ge t$ of $\text{TTC}(P)$ where agent~$i$ is involved in a trading cycle, she is assigned an object weakly better than the object she relinquishes from her endowment. In other words, there is a bijection $\sigma: \omega_i \cap O^t  \to \varphi^{\text{TTC}}_i(P) \cap O^t $ such that, for each $o \in \omega_i \cap O^t $, $  \sigma(o) \mathrel{R_i} o$.

Together with (\ref{plusone}) and (\ref{plustwo}), we conclude that there is a bijection $\pi : \varphi_i^{\text{TTC}}(P') \to \varphi_i^{\text{TTC}}(P)$ such that, for each $o\in \varphi^{\text{TTC}}_i(P')$, $\pi(o) \mathrel{R_i} o$. Since $P_i$ is responsive, we must have $\varphi^{\text{TTC}}_i(P) \mathrel{R_i} \varphi^{\text{TTC}}_i(P')$.\footnote{Starting from $\varphi^{\text{TTC}}_i(P')$, replace each object $o \in \varphi^{\text{TTC}}_i(P')$ with object $\pi(o)$, one at a time, and apply the definition of responsiveness.}
\hfill{}$\blacksquare$
\end{casenv}

\subsection{Proof of Proposition~\ref{prop:NOM}}

Let $P \in \mathcal{R}$, $i \in N$, and $P'_i \in \mathcal{R}_i$ be such that $\varphi^{\mathrm{TTC}}_i(P'_i, P_{-i}) \mathrel{P_i} \varphi^{\mathrm{TTC}}_i(P)$.

\begin{itemize}
    \item[(i)] Let $P_i|_O : o_1, \dots, o_{|\omega_i|}, \dots$. In particular, $P_i|_O$ ranks every object in $\{o_1, \dots, o_{|\omega_i|}\}$ above every object in $O \setminus \{o_1, \dots, o_{|\omega_i|}\}$. For each $j \in N \setminus \{i\}$, let $P^*_j|_O:\dots,o_1, \dots, o_{|\omega_i|}$. In particular, $P^*_j|_O$ ranks every object in $O \setminus \{o_1, \dots, o_{|\omega_i|}\}$ above every object in $\{o_1, \dots, o_{|\omega_i|}\}$. Then $\varphi^{\mathrm{TTC}}_i(P_i, P^*_{-i}) = \{o_1, \dots, o_{|\omega_i|}\}$ (otherwise $\varphi^{\mathrm{TTC}}$ would not satisfy \IGE). Since $P_i$ is responsive, $\{o_1, \dots, o_{|\omega_i|}\}$ is the most-preferred $|\omega_i|$-subset of $O$ according to $P_i$. Consequently, 
    $B_{P_i} \left(
    \mathcal{O}_i(P_i)
    \right) \mathrel{R_i}
    B_{P_i} \left(
    \mathcal{O}_i(P'_i)
    \right)$.
    
    \item[(ii)] For each $j \in N \setminus \{i\}$, let $\omega_j =\{o^j_1 ,\dots ,o^j_{|\omega_j|}\}$ and $P^*_j|_O:o^j_1, \dots, o^j_{|\omega_j|},\dots$. In particular, $P^*_j|_O$ ranks every object in $\omega_j$ above every object in $O \setminus \omega_j$. Then $\varphi_i^{\mathrm{TTC}}(P_i, P^*_{-i}) = \varphi_i^{\mathrm{TTC}}(P'_i, P^*_{-i}) = \omega_i$. Consequently, since $\varphi^{\mathrm{TTC}}$ satisfies \IR on $\mathcal{R}$, we have 
    $W_{P_i} \left(
    \mathcal{O}_i(P_i)
    \right) = \omega_i \mathrel{R_i} W_{P_i} \left(
    \mathcal{O}_i(P'_i)
    \right)$.    \hfill{}$\blacksquare$
\end{itemize}

\subsection{Proof of Proposition~\ref{ATTC is drop SP}}

Fix an agent $i \in N$. To prove that $i$ cannot manipulate using drop strategies, we consider an alternative procedure that proceeds exactly as in the ATTC algorithm except that, at each step in which multiple trading cycles arise, it executes only the cycles that do not involve $i$. Thus any trading cycle involving agent~$i$ is temporarily deferred, and it is executed only when it is the unique remaining cycle. It is easy to see that the order in which cycles are executed does not affect the final allocation; hence this modified procedure yields the same allocation as the original ATTC algorithm.\footnote{A similar modification is used by \citet{gonczarowski2023} to ``expose'' the strategy-proofness of TTC in single-object environments.}

Formally, given a preference profile $P \in \mathcal{CL}$, the allocation $\varphi^\mathrm{ATTC}(P)$ is equivalently obtained via the following algorithm, which we denote $\mathrm{ATTC}^i(P)$. The algorithm defines, for each step $t$, several ingredients that will be used throughout the proof: (i) the remaining set of objects $O^t(P)$, (ii) the directed graph $G^t(P)$, (iii) the full set $\widetilde{\mathcal{C}}_t(P)$ of cycles that \emph{arise} at step $t$, (iv) the subset $\overline{\mathcal{C}}_t(P)$ of trading cycles that are \emph{executed} at step $t$, and (v) the partial allocation $\mu^t(P)$ constructed by the end of step $t$.

\noindent \rule{1\columnwidth}{1pt}
\paragraph{Algorithm:} $\mathrm{ATTC}^i(P)$
\begin{description}
    \item [\textnormal{\emph{Initialization}:}] Set $\mu^{0}(P) \coloneqq\left(\emptyset\right)_{j\in N}$
    and $O^{1}(P)\coloneqq O$.
    \item [\textnormal{\emph{Step~$t \geq 1$}:}] ~
        Construct a directed graph $G^t(P)$ on $N \cup O^t(P)$ as follows. Each agent~$j$ points to $\max_{P_j}(O^t(P) \mid \mu_j^{t-1}(P))$, and each object in $O^t(P)$ points to its owner. Let $\widetilde{\mathcal{C}}_t(P)$ be the resulting set of trading cycles, and let
        $$
        \overline{\mathcal{C}}_t(P) = \begin{cases}
            \{C \in \widetilde{\mathcal{C}}_t(P) \mid i \notin N(C) \}, & \text{if } |\widetilde{\mathcal{C}}_t(P)| \geq 2 \\
            \widetilde{\mathcal{C}}_t(P), & \text{if }|\widetilde{\mathcal{C}}_t(P)| = 1.
        \end{cases}
        $$
        Execute all cycles in $\overline{\mathcal{C}}_t(P)$, yielding the partial allocation $\mu^t(P)$ with
        $$
        \mu^t_j(P) = \begin{cases}
            \mu^{t-1}_j(P) \cup \{\max_{P_j}(O^t(P) \mid \mu^{t-1}_j(P)) \}, & \text{if } j \in \bigcup_{C \in \overline{\mathcal{C}}_t(P)} N(C) \\
            \mu^{t-1}_j(P), & \text{otherwise.}
        \end{cases}
        $$
        Remove all objects involved in executed trading cycles to form the set $O^{t+1}(P) = O^t(P) \setminus \bigcup_{C \in \overline{\mathcal{C}}_t(P)} O(C)$. Proceed to step~$t+1$ if $O^{t+1}(P) \neq \emptyset$; proceed to Termination otherwise.
        
    \item [\textnormal{\emph{Termination}:}] Let $T$ be the earliest step with $O^{T+1}(P) = \emptyset$. Return $\varphi^{\text{ATTC}}\left(P\right)\coloneqq\mu^{T}(P)$.
\end{description}
\rule{1\columnwidth}{1pt}

Let $P \in \mathcal{CL}$. Let $\tilde{P}_i \in \mathcal{D}_i(P_i)$ be obtained from $P_i$ by dropping an object $x \in O \setminus \omega_i$, and denote $\tilde{P} \coloneqq (\tilde{P}_i, P_{-i})$. We prove that $\varphi^{\mathrm{ATTC}}_i(P) \mathrel{R_i} \varphi^{\mathrm{ATTC}}_i(\tilde{P})$ by comparing the two runs of the modified ATTC algorithm.\footnote{A direct step-by-step comparison of $\mathrm{ATTC}(P)$ and $\mathrm{ATTC}(\tilde{P})$ is awkward because agent~$i$ may receive the same object at different steps in the two runs, creating a timing mismatch. The modified algorithm is used to synchronize the two runs $\mathrm{ATTC}^i(P)$ and $\mathrm{ATTC}^i(\tilde{P})$ up to the earliest step at which $i$ receives a different object.}

If $\overline{\mathcal{C}}_t(P) = \overline{\mathcal{C}}_t(\tilde{P})$ for every step $t$, then the two runs execute exactly the same cycles; hence $\varphi^{\mathrm{ATTC}}(P) = \varphi^{\mathrm{ATTC}}(\tilde{P})$ and the claim holds.

Otherwise, let $t$ be the earliest step such that $\overline{\mathcal{C}}_t(P) \neq \overline{\mathcal{C}}_t(\tilde{P})$. Then $\overline{\mathcal{C}}_s(P) = \overline{\mathcal{C}}_s(\tilde{P})$ for each step $s < t$, which implies that the partial allocations and the remaining objects coincide at the beginning of step $t$ in the two runs, i.e.,
$$
\mu^{t-1}(P) = \mu^{t-1}(\tilde{P}) \quad \text{and} \quad O^t(P) = O^t(\tilde{P}).
$$

Consider the directed graphs $G^t(P)$ and $G^t(\tilde{P})$ constructed at step~$t$ in the two runs. Because $P_{-i} = \tilde{P}_{-i}$, $\mu^{t-1}(P) = \mu^{t-1}(\tilde{P})$, and $O^t(P) = O^t(\tilde{P})$, every agent $j \in N \setminus \{i\}$ points to the same object in both graphs. Hence, the sets of cycles not involving $i$ coincide:
\begin{equation}\label{eq:cyles without i coincide}
\{C \in \widetilde{\mathcal{C}}_t(P) \mid i \notin N(C)\} = \{C \in \widetilde{\mathcal{C}}_t(\tilde{P}) \mid i \notin N(C)\}.
\end{equation}
If the common set in (\ref{eq:cyles without i coincide}) were nonempty, then by definition of $\mathrm{ATTC}^i(P)$ and $\mathrm{ATTC}^i(\tilde{P})$ we would have 
$$
\overline{\mathcal{C}}_t(P) = \{C \in \widetilde{\mathcal{C}}_t(P) \mid i \notin N(C)\} = \{C \in \widetilde{\mathcal{C}}_t(\tilde{P}) \mid i \notin N(C)\} = \overline{\mathcal{C}}_t(\tilde{P}),$$
contradicting the choice of $t$. Thus
$$
\{C \in \widetilde{\mathcal{C}}_t(P) \mid i \notin N(C)\} = \{C \in \widetilde{\mathcal{C}}_t(\tilde{P}) \mid i \notin N(C)\} = \emptyset.
$$

Since agent~$i$ can be involved in at most one cycle in $G^t(P)$ and in $G^t(\tilde{P})$, and each graph contains a cycle, we must have 
$\overline{\mathcal{C}}_t(P) = \widetilde{\mathcal{C}}_t(P) = \{C\}$ and $\overline{\mathcal{C}}_t(\tilde{P}) = \widetilde{\mathcal{C}}_t(\tilde{P}) = \{\tilde{C}\}$
for some trading cycles $C, \tilde{C}$ with $i \in N(C) \cap N(\tilde{C})$. Moreover, $C \neq \tilde{C}$ because $\overline{\mathcal{C}}_t(P) \neq \overline{\mathcal{C}}_t(\tilde{P})$. Thus agent~$i$ must point to different objects in $G^t(P)$ and $G^t(\tilde{P})$ (since all other agents point to the same object). Because $Y \coloneqq \mu^{t-1}_i(P) = \mu^{t-1}_i(\tilde{P})$ and $\tilde{P}_i$ is obtained from $P_i$ by dropping object $x$, the conditional marginal preferences $P_i(Y)|_O$ and $\tilde{P}_i(Y)|_O$ agree on $O \setminus \{x\}$; hence if $i$ points to different objects in the two graphs, then she necessarily points to $x$ in $C$ but not in $\tilde{C}$.
Since $C$ is the unique cycle executed at step~$t$ of $\mathrm{ATTC}^i(P)$, we have $x \in \varphi^{\mathrm{ATTC}}_i(P)$.

Because $x \notin \omega_i$ and $\tilde{P}_i$ is obtained from $P_i$ by dropping $x$, the object $x$ is ranked below $\min_{\tilde{P}_i(Z)}(\omega_i)$ according to every conditional marginal preference $\tilde{P}_i(Z)|_O$ associated with $\tilde{P}_i$. Since $\varphi^{\mathrm{ATTC}}$ satisfies \WELB, we must have $x \notin \varphi^{\mathrm{ATTC}}_i(\tilde{P})$.

Finally, since agent~$i$ points to $x$ at step~$t$ of $\mathrm{ATTC}^i(P)$, we have $
x = \max_{P_i} \left(O^t(P) \mid \mu^{t-1}_i(P) \right)$.
Since $P_i$ is conditionally lexicographic and $x \in \varphi^{\mathrm{ATTC}}_i(P) \setminus \varphi^{\mathrm{ATTC}}_i(\tilde{P})$, agent~$i$ strictly prefers to complete her partial assignment $\mu^{t-1}_i(P)$ by adding the set $\varphi^{\mathrm{ATTC}}_i(P) \cap O^t(P)$ instead of $\varphi^{\mathrm{ATTC}}_i(\tilde{P}) \cap O^t(\tilde{P})$. Since 
\begin{align*}
    \varphi^{\mathrm{ATTC}}_i(P) & 
    = \mu^{t-1}_i(P) \cup  \left(\varphi^{\mathrm{ATTC}}_i(P) \cap O^t(P)\right) \\ 
    \text{and} \quad 
    \varphi^{\mathrm{ATTC}}_i(\tilde{P}) & = \mu^{t-1}_i(P) \cup  \left(\varphi^{\mathrm{ATTC}}_i(\tilde{P}) \cap O^t(\tilde{P})\right),
\end{align*} 
we conclude that
$\varphi^{\mathrm{ATTC}}_i(P) \mathrel{P_i} \varphi^{\mathrm{ATTC}}_i(\tilde{P})$.\hfill{}$\blacksquare$




\subsection{Proof of Theorem~\ref{characterization of ATTC}}

Toward contradiction, suppose that $\varphi$ satisfies the properties but $\varphi \neq \varphi^{\text{ATTC}}$. The argument here mirrors the proof of Theorem~\ref{Multi-unit characterization with TP} for the lexicographic domain (see the sketch in Section~\ref{results for lexicographic preferences}).

For each $P \in \mathcal{CL}$ and each $t \in \mathbb{N}$, recall that $\mathcal{C}_t(P)$ denotes the set of trading cycles that arise at step~$t$ of $\text{ATTC}\left(P\right)$. (We assume that $\mathcal{C}_t(P) = \emptyset$ if $\text{ATTC}(P)$ terminates before step $t$.) The ``size'' function $s:\mathcal{CL} \to \mathbb{N}$ is defined, for each $P \in \mathcal{CL}$, by
$$
s(P) = \sum_{i\in N}\sum_{Y\in2^{O}}\left|\{o \in O \mid o \mathrel{R_i(Y)} \min_{P_i(Y)}(\omega_i)\}\right|,
$$
The ``similarity'' function $\rho: \mathcal{CL} \to \mathbb{N} \cup \{\infty\}$ is such that, for each $P \in \mathcal{CL}$, 
$$
\rho(P) = \inf \{t \in \mathbb{N} \mid \varphi(P) \text{ does not execute each trading cycle in } \mathcal{C}_t(P)\},$$
where it is understood that $\rho(P) = \infty$ precisely when $\varphi(P) = \varphi^{\text{ATTC}}(P)$.

Among all profiles that minimize the similarity function $\rho$, let $P$ be one that further minimizes the size function $s$. Thus, for any profile $P' \in \mathcal{CL}$, either (i) $\rho(P)<\rho(P')$, or (ii) $\rho(P)=\rho(P')$ and $s(P) \leq s(P')$.

Let $t \coloneqq \rho(P)$. Then $\varphi(P)$ executes each trading cycle in $\bigcup_{\tau = 1}^{t-1} \mathcal{C}_\tau(P)$, but it does not execute some trading cycle $C \in \mathcal{C}_t(P)$. Suppose that
\[
C=\left(o_1, i_{1},o_2,i_{2},\dots,o_k,i_{k},o_{k+1} = o_1\right).
\]

Because $\varphi(P)$ does not execute $C$, there is an agent $i_\ell \in N(C)$ who is not assigned the object she points to within $C$, i.e., $o_{\ell+1} \notin \varphi_{i_\ell}(P)$. Without loss of generality, let $i_\ell = i_k$. Thus, $o_{k+1} = o_{1} \notin \varphi_{i_k}(P)$. 

Let $O^t$ denote the set of objects available at the beginning of Step $t$ of $\text{ATTC}(P)$, i.e., $O^t \coloneqq O \setminus \left(\bigcup_{\tau = 1}^{t-1} \bigcup_{C' \in \mathcal{C}_\tau(P)} O(C')\right)$. Then $\varphi_{i_k}^{\text{ATTC}}(P) \setminus O^t$ is the set of objects assigned to agent~$i_k$ before step~$t$ of $\text{ATTC}(P)$; moreover $\varphi_{i_k}(P) \setminus O^t = \varphi^{\text{ATTC}}_{i_k}(P) \setminus O^t$.

\begin{claim}
$\varphi_{i_k}(P) \cap O^t = \omega_{i_k} \cap O^t$. \label{claim22}
\end{claim}
\begin{proof} 
Suppose otherwise. By \BAL and the fact that $\left|\varphi_{i_{k}}\left(P\right)\setminus O^{t}\right|=\left|\omega_{i_{k}} \setminus O^{t}\right|$,
we must have $\left|\varphi_{i_{k}}\left(P\right)\cap O^{t}\right|=\left|\omega_{i_{k}}\cap O^{t}\right|$.
Consequently, $\varphi_{i_{k}}\left(P\right)\cap O^{t}\neq\omega_{i_{k}}\cap O^{t}$
implies there is an object $o' \in  \varphi_{i_k}(P) \cap O^t$ with $o' \not \in \omega_{i_k}$.

By the definition of $\text{ATTC}(P)$, we know that $o_{1} = \max_{P_{i_k}}(O^t \mid \varphi_{i_k}(P) \setminus O^t)$. Because $o' \in \varphi_{i_k}(P) \cap O^t$ and $o_{1} \notin \varphi_{i_k}(P)$, it follows that $o_1 \mathrel{P_{i_k}(\varphi_{i_k}(P) \setminus O^t)} o'$ (thus $o_1 \neq o'$). Furthermore, \WELB implies that $o' \mathrel{P_{i_k}(\varphi_{i_k}(P))} \min_{P_{i_k}(\varphi_{i_k}(P))}(\omega_{i_k})$.
 
 Let $P'_{i_k}$ be the drop strategy obtained from $P_{i_k}$ by dropping object $o'$.  Let $P' \coloneqq (P'_{i_k}, P_{-i_k})$. Then, for each $Y \subseteq O$, 
 $$
 \left|\{o \in O \mid o \mathrel{R'_i(Y)} \min_{P'_i(Y)}(\omega_i)\}\right| \leq \left|\{o \in O \mid o \mathrel{R_i(Y)} \min_{P_i(Y)}(\omega_i)\}\right|,
 $$
 with strict inequality for $Y = \varphi_{i_k}(P)$ because $o' \mathrel{P_{i_k}(\varphi_{i_k}(P))} \min_{P_{i_k}(\varphi_{i_k}(P))}(\omega_{i_k})$. Thus $s(P') < s(P)$. Consequently, the choice of $P$ implies that $\rho(P') > \rho(P) = t$.

By definition of $P'_{i_k}$, the two conditional marginal preferences $P'_{i_k}(\varphi_{i_k}(P) \setminus O^t)$ and $P_{i_k}(\varphi_{i_k}(P) \setminus O^t)$ agree for all objects at least as good as $o_1$. In particular, $\max_{P'_{i_k}}(X \mid Y) = \max_{P_{i_k}}(X \mid Y)$ whenever $Y \subseteq \varphi_{i_k}(P) \setminus O^t$ and $X \supseteq O^t$. Consequently, $\text{ATTC}(P)$ and $\text{ATTC}(P')$ execute precisely the same trading cycles at every step $\tau \leq t$, i.e., $\mathcal{C}_\tau(P)= \mathcal{C}_\tau(P')$ for $\tau=1,\ldots,t$. 
Thus, $\rho(P') > t$ implies that $\varphi(P')$ executes all trading cycles in $\bigcup_{\tau = 1}^{t}\mathcal{C}_\tau(P)$. In particular,
$$(\varphi_{i_k}(P) \setminus O^t) \cup \{o_{1}\} \subseteq \varphi_{i_k}(P').$$
Since $P_{i_k}$ is conditionally lexicographic and $\max_{P_{i_k}}(O^t \mid \varphi_{i_k}(P) \setminus O^t) = o_{1} \notin \varphi_{i_k}(P)$, Definition~\ref{characterization of CL preferences} implies that $\varphi_{i_k}(P') \mathbin{P_{i_k}} \varphi_{i_k}(P)$. This contradicts \DSP.
\end{proof}

Next, we show that at least two agents are involved in $C$.

\begin{claim}\label{claim33}
	$|N(C)|>1$.
\end{claim}
\begin{proof}
Toward contradiction, suppose that $|N(C)|=1$. Thus  $i_1 = i_k$ and $C = (o_{1}, i_1, o_{1})$. Then $o_{1} \in \omega_{i_k}$ and Claim~\ref{claim22} implies that $o_{1} \in \varphi_{i_k}(P)$, a contradiction.
\end{proof}

By Claim~\ref{claim22} and Claim~\ref{claim33}, agent~$i_{k-1}$ points to $o_{k}$ on $C$, yet she does not receive $o_{k}$. An argument similar to the proof of Claim~\ref{claim22} implies that $\varphi_{i_{k-1}}(P) \cap O^t = \omega_{i_{k-1}} \cap O^t$. 

Proceeding by induction, we conclude that for each agent~$i_\ell$ involved in $C$, we have $\varphi_{i_\ell}(P) \cap O^t = \omega_{i_\ell} \cap O^t$. Thus, $C$ is a Pareto-improving trading cycle at $P$. It follows that $\varphi$ is not \IGE, a contradiction.
\hfill{}$\blacksquare$

\subsection{Proof of Proposition~\ref{maximal domain 1}}
Without loss of generality, $P_1 \in \mathcal{M}_1 \setminus \mathcal{CL}_1$. Since $P_1$ is not conditionally lexicographic, there exist disjoint subsets $X, Y \in 2^O$ with $X \neq \emptyset$ such that
$$\text{for all }x \in X \text{, there exists }Z_x \subseteq X \setminus x \text{ such that }(Y \cup Z_x) \mathrel{P_1}(Y \cup x).$$
Because $P_1$ is \defterm{monotone}, $Z_x \subseteq X \setminus x$ implies $ [Y \cup (X \setminus x)] \mathrel{R_1} (Y \cup Z_x) $.
Thus, we have
$$\text{for all }x \in X,\quad [Y \cup (X \setminus x)] \mathrel{P_1} (Y \cup x).$$
Note also that $|X| \geq 3.$\footnote{If $X$ is a singleton, say $X = \{x\}$, then $Y \mathrel{P_1} (Y \cup x)$, a violation of monotonicity. If $X$ contains two objects, say $X = \{x, y\}$, we would have $(Y \cup x) \mathrel{P_1} (Y \cup y) \mathrel{P_1} (Y \cup x)$, a violation of transitivity.}

Let $x^* \in X$ be such that $(Y \cup x^*) \mathrel{R_1} (Y \cup x)$ for all $x \in X$.
Then, since $|X| \geq 3$ and $P_1$ is monotonic, it follows that
$$\text{for all }x \in X,\quad [Y \cup (X \setminus x)] \mathrel{P_1} (Y \cup x^*) \mathrel{R_1} (Y \cup x).$$
\begin{casenv}
\item Suppose $X \cup Y = O$. Let $N = \{1, 2\}$, and let $P_2 \in \mathcal{L}$ 
be such that (i) $\max_{P_2}(O) = x^*$, and (ii) for all $x \in X$, $x \mathrel{P_2} Y$. Then, in particular, 
$$x^* \mathrel{P_2} [Y \cup (X \setminus x^*)] \mathrel{R_2} (X \setminus x^*).$$
Consider the allocations
$$\mu\coloneqq (Y\cup x^{*},X\setminus x^{*})\quad\text{and}\quad\overline{\mu}\coloneqq (Y\cup\left(X\setminus x^{*}\right),x^{*}).$$
Then $\mu$ is not \PE at $P$ because it is Pareto dominated by $\overline{\mu}$ at $P$.

We claim that $\mu$ satisfies \IGE at $P$. Indeed, consider any trading cycle $C = (x, 2, y, 1, x)$, where $x \in \mu_2$ and $y \in \mu_1$. Executing $C$ yields the allocation $\mu'$, where
$$\mu'_1=\left(Y\cup\left\{ x^{*},x\right\} \right)\setminus y \quad\text{and}\quad \mu'_2 = \left(X\setminus\left\{ x^{*},x\right\} \right)\cup y.$$
If $y \in Y$, then this exchange harms agent 2, i.e., $\mu_2 \mathrel{P_2} \mu'_2$. On the other hand, if $y = x^*$, then this exchange harms agent 1 because $\mu_1 = (Y \cup x^*) \mathrel{P_1} (Y \cup x) = \mu'_1$. Hence, $\mu$ satisfies \IGE.
\item Suppose $X \cup Y \subsetneq O$. Let $N = \{1, 2, 3\}$ and denote $\overline{O} \coloneqq O \setminus (X \cup Y)$. Let $P_2 \in \mathcal{L}$ be such that (i) $\max_{P_2}(O) = x^*$, and (ii) for all $x \in X$, $x \mathrel{P_2} (O \setminus X)$. Let $P_3 \in \mathcal{L}$ be such that (i) for all $z \in \overline{O}$, $z \mathrel{P_3} (X \cup Y)$.

Consider the allocations
$$\mu\coloneqq (Y\cup x^{*},X\setminus x^{*},\overline{O}) \quad\text{and}\quad \overline{\mu}\coloneqq (Y\cup\left(X\setminus x^{*}\right),x^{*},\overline{O} ).$$
Then $\mu$ is not \PE at $P$ because it is Pareto dominated by $\overline{\mu}$ at $P$.

We claim that $\mu$ satisfies \IGE at $P$. Because $\mu_3$ is agent 3's top-ranked $|\overline{O}|$-subset of $O$, any trading cycle at $\mu$ involving agent 3 must harm agent 3. Thus, any Pareto-improving trading cycle at $\mu$ must involve only agents 1 and 2. However, there is no such exchange by the argument in Case 1. Hence, $\mu$ satisfies \IGE. \hfill{}$\blacksquare$
\end{casenv}

\section{Examples}\label{Appendix: independence}

The examples below illustrate the logical independence of the properties used in Theorems~\ref{Multi-unit characterization with TP}--\ref{responsive characterization with IR} (see Tables~\ref{tab:properties_L} and~\ref{tab:properties_R} for a summary of which properties each example satisfies on $\mathcal{L}$ and $\mathcal{R}$). Since $\mathcal{L} \subseteq \mathcal{R}$, any rule defined on $\mathcal{R}$ is also viewed as a rule on $\mathcal{L}$ by restriction.

Concretely, for each property appearing in one of these theorems, we exhibit a rule (defined on the relevant domain) that violates this property while satisfying all other properties in the theorem. To keep notation light, some examples fix convenient values for the set $N$ of agents and the endowment $\omega$. Analogous constructions can be given on $\mathcal{CL}$; we omit them for brevity.

\begin{example}\label{No-trade rule} \defterm{No-trade rule} (fails \IGE; satisfies \MAR, \BAL, \WELB, \IR, and \SP)

Let $\varphi^{\text{NT}}$ be the \defterm{no-trade rule} on $\mathcal{R}$, which selects $\omega$ for each problem. Then $\varphi^{\text{NT}}$ satisfies \MAR, \BAL, \WELB, \IR, and \SP on $\mathcal{R}$, but it fails \IGE on  $\mathcal{L}$ (and thus on $\mathcal{R}$).
\hfill{} $\diamond$
\end{example}

\begin{example}\label{Balanced SD} \defterm{Balanced Serial Dictatorship} (fails \WELB and \IR; satisfies \MAR, \BAL, \IGE, and \SP)

Let $\varphi^{\mathrm{BSD}}$ be the \defterm{Balanced Serial Dictatorship} on $\mathcal{R}$, defined recursively as follows: for each $P \in \mathcal{R}$, $\varphi^{\mathrm{BSD}}_1(P) = \max_{P_1}(O, |\omega_1|)$ and, for each $i = 2, \dots, n$,
$$\varphi^{\mathrm{BSD}}_i(P) = \max_{P_i}\left(O \setminus \bigcup_{j = 1}^{i-1} \varphi^{\mathrm{BSD}}_j(P), |\omega_i|\right),$$
where $\max_{P_i}(X, k)$ denotes the most-preferred $k$-subset of $X$ at $P_i$. Then $\varphi^{\mathrm{BSD}}$ satisfies \MAR, \BAL, \IGE, and \SP on $\mathcal{R}$, but it fails \WELB and \IR on $\mathcal{L}$ (and thus on $\mathcal{R}$).
\hfill{} $\diamond$
\end{example}

\begin{example}\label{Deviation from TTC} \defterm{A rule $\varphi^{\neg \TP}$} (fails \TP and \DSP; satisfies \MAR, \BAL, \IGE, \WELB, and \IR)

Let $N = \{1, 2\}$, $O = \{a, b, c\}$, and $\omega = (\{a,b\}, \{c\})$. Fix $P^* \in \mathcal{L}$ with $P^*_1|_O : c, a, b$ and $P^*_2|_O: a, b, c$, and let $\mu^* = (\{a, c\}, \{b\})$. Define the rule $\varphi^{\neg \TP}$ on $\mathcal{R}$ by
$$
\varphi^{\neg \TP}(P) =
\begin{cases}
    \mu^*, & \text{if } P|_O = P^*|_O \\
    \varphi^{\mathrm{TTC}}(P), & \text{otherwise}.
\end{cases}
$$
It is easy to see that $\varphi^{\neg \TP}$ satisfies \MAR, \BAL, \IGE, \WELB, and \IR on $\mathcal{R}$. 

We show that $\varphi^{\neg \TP}$ fails \TP and \DSP on $\mathcal{L}$. Consider $P'_2 \in \mathcal{L}_2$ with $P'_2|_O : a, c, b$. Then $\varphi^{\neg \TP}_2(P^*_1, P'_2) = \{a\} \mathrel{P^*_2} \{b\} = \varphi^{\neg \TP}_2(P^*)$, so $P'_2$ is a profitable manipulation for agent~$2$ at $P^*$. Since $P'_2$ is both a truncation strategy and a drop strategy for $P^*_2$ (it is obtained by dropping the singleton tail subset $\{o \in O \setminus \omega_2 \mid b \mathrel{R^*_2} o\} = \{b\}$), $\varphi^{\neg \TP}$ fails \TP and \DSP on $\mathcal{L}$  (and thus on $\mathcal{R}$).
\hfill{} $\diamond$
\end{example}

\begin{example}\label{NOT BAL} \defterm{A rule $\varphi^{\neg \BAL}$ on $\mathcal{L}$} (fails \BAL; satisfies \PE, \WELB, \IR, and \SDSP)

Let $N = \{1, 2\}$, $O = \{a, b, c\}$, and $\omega = (\{a, b\}, \{c\})$. Consider the lexicographic profile $P^*$ with $P_1^*: c, a, b$ and $P^*_2 : a, b, c$, and let $\mu^* \coloneqq (\{c\}, \{a,b\})$. Define the rule $\varphi^{\neg \BAL}$ on $\mathcal{L}$ by
$$
\varphi^{\neg \BAL}(P) =
\begin{cases}
    \mu^*, & \text{if } P = P^* \\
    \varphi^{\mathrm{TTC}}(P), & \text{otherwise}.
\end{cases}
$$
It is easy to see that $\varphi^{\neg \BAL}$ satisfies \PE, \WELB, and \IR. 

We show that $\varphi^{\neg \BAL}$ satisfies \DSP (and thus \SDSP by Lemma~\ref{Proposition: drop SP + WELB implies subset-drop SP}).
Fix $P \in \mathcal{L}$, agent $i \in N$, and a drop strategy $P'_i \in \mathcal{D}_i(P_i) \setminus \{P_i\}$. Denote $P' \coloneqq (P'_i, P_{-i})$ and consider three cases.

\begin{casenv}
    \item Suppose $P \neq P^*$ and $P' \neq P^*$. Then $\varphi^{\neg \BAL}(P) = \varphi^{\mathrm{TTC}}(P)$ and $\varphi^{\neg \BAL}(P') = \varphi^{\mathrm{TTC}}(P')$, so the deviation is not profitable because $\varphi^{\mathrm{TTC}}$ satisfies \SDSP on $\mathcal{L}$.
    \item Suppose $P = P^*$. If $i = 2$, then at $P_2$ the assignment $\varphi^{\neg \BAL}(P) = \mu_2^*$ is agent 2's most-preferred feasible bundle (i.e., among all nonempty, proper subsets of $O$), so no deviation can be profitable. If $i = 1$, then any proper drop strategy for $P_1$ must drop object $c \in O \setminus \omega_1$, giving $P'_1 : a, b, c$. Thus $\varphi^{\neg \BAL}_1(P) = \{c\} \mathrel{P_1} \{a, b\} = \varphi^{\neg \BAL}_1(P')$.
    \item The case $P' = P^*$ is impossible. Indeed, under $P^*_1$ the bottom object is $b \in \omega_1$, so $P^*_1$ is not a proper drop strategy for any $P_1 \in \mathcal{L}$. Similarly, under $P^*_2$ the bottom object is $c \in \omega_2$, so $P^*_2$ is not a proper drop strategy for any $P_2 \in \mathcal{L}$.
    \hfill{} $\diamond$
\end{casenv}
\end{example}

\begin{example}\defterm{A rule $\varphi^{\neg \WELB}$ on $\mathcal{L}$} (fails \WELB; satisfies \BAL, \IGE, \IR, and \SDSP)
\label{no characterization with IR}

Let $N = \{1, 2\}$, $O = \{a, b, c, d\}$, and $\omega = (\{a, b\}, \{c, d\})$. Consider the lexicographic profile $P^*$ with $P^*_1 : c, a, b, d$ and $P^*_2 : a, b , c, d$, and let $\mu^* = (\{c, d\}, \{a, b\})$. Define $\varphi^{\neg \WELB}$ on $\mathcal{L}$ by
$$
\varphi^{\neg \WELB}(P) = \begin{cases}
    \mu^*, & \text{if } P = P^* \\
    \varphi^{\mathrm{TTC}}(P), & \text{otherwise}.
\end{cases}
$$
It is easy to see that $\varphi^{\neg \WELB}$ satisfies \BAL, \IGE, and \IR. Furthermore,  $\varphi^{\neg \WELB}$ fails \WELB at $P^*$ because $d \in \varphi_1^{\neg \WELB}(P^*)$ although $\min_{P_1^*}(\omega_1) = b \mathrel{P^*_1} d$. 

We now show that $\varphi^{\neg \WELB}$ satisfies \SDSP (hence \TP and \DSP). Fix $P \in \mathcal{L}$, agent $i \in N$, and a subset-drop strategy $P'_i \in \mathcal{S}_i(P_i) \setminus \{P_i\}$. Denote $P' \coloneqq (P'_i, P_{-i})$ and consider three cases.

\begin{casenv}
    \item Suppose $P = P^*$. If $i = 2$, then $\varphi^{\neg \WELB}_2(P) = \mu^*_2$ is agent~$2$'s most-preferred two-object bundle at $P^*_2$, so no deviation can be profitable. If $i = 1$, then any proper subset-drop strategy for $P_1$ must drop object $c \in O \setminus \omega_1$ (possibly together with $d$), which means that $P'_1$ is either $P^1_1:a, b, d, c$ or $P^2_1:a, b, c, d$. In each case, $\varphi^{\mathrm{TTC}}$ assigns $\varphi^{\mathrm{TTC}}(P') = \omega$. Thus $\varphi^{\neg \WELB}_1(P) = \{c,d\} \mathrel{P_1} \{a, b\} = \varphi^{\neg \WELB}_1(P')$.
    \item Suppose $P' = P^*$. Then necessarily $i = 1$ and $P_2 = P^*_2$ (since $P^*_2$ ranks an object $d \in \omega_2$ at the very bottom, it is not a proper subset-drop strategy for any $P_2 \in \mathcal{L}_2$).
    Moreover, $P'_1 = P^*_1$ can be obtained from $P_1$ only by dropping object $d$, which means that $P_1$ is either $P^1_1 :d, c, a, b$, $P^2_1: c, d, a, b$, or $P^3_1:c, a, d, b$. In each case, $\varphi^{\mathrm{TTC}}$ assigns $\varphi^{\mathrm{TTC}}(P) = \mu^* = \varphi^{\neg \WELB}(P')$, so the deviation is not profitable.
    \item Suppose $P \neq P^*$ and $P' \neq P^*$. Then $\varphi^{\neg \WELB}(P) = \varphi^{\mathrm{TTC}}(P)$ and $\varphi^{\neg \WELB}(P') = \varphi^{\mathrm{TTC}}(P')$, so the deviation is not profitable because $\varphi^{\mathrm{TTC}}$ satisfies \SDSP on $\mathcal{L}$.
    \hfill $\diamond$
\end{casenv}
\end{example}

\begin{example}\label{non-marginal}
\defterm{A rule $\varphi^{\neg \MAR}$ on $\mathcal{R}$} (fails \MAR; satisfies \IGE, \WELB, \IR, and \TP)

Let $N = \{1, 2\}$ and $\omega = (\{a, b, d\}, \{c, e\})$. Fix an allocation $\mu^* \coloneqq (\{b, c, e\}, \{a, d\})$. Let 
$$
\mathcal{R}^* = \{P^* \in \mathcal{R} \mid P^*_1|_O:a,e,b,c,d \text{ and } P^*_2|_O:a,c,b,d,e\}.
$$
Since $\varphi^{\mathrm{TTC}}$ is marginal, for each $P \in \mathcal{R}^*$ we have $\varphi^{\mathrm{TTC}}(P) = (\{a, d, e\}, \{b, c\})$.

Define the rule $\varphi^{\neg \MAR}$ on $\mathcal{R}$ by
$$
\varphi^{\neg \MAR}(P) = 
\begin{cases}
\mu^*, & \text{if } P \in \mathcal{R}^* \text{ and } \mu^* \text{ Pareto dominates } \varphi^{\mathrm{TTC}}(P) =(\{a, d, e\}, \{b, c\}), \\
\varphi^{\mathrm{TTC}}(P), & \text{otherwise.} 
\end{cases}
$$
One easily checks that $\varphi^{\neg \MAR}$ satisfies \IGE, \WELB, and \IR (and \BAL). 

We show that $\varphi^{\neg \MAR}$ satisfies \TP. Fix $P \in \mathcal{R}$ and $i \in N$. Let $P'_i \in \mathcal{T}_i(P_i)$ and denote $P' \coloneqq (P'_i, P_{-i})$. Denote the unique agent in $N \setminus \{i\}$ by $j$. We consider two cases.

\begin{casenv}
\item Suppose $P \in \mathcal{R}^*$. Consider two further subcases.
    \begin{casenv}
        \item Suppose $\varphi_j^{\mathrm{TTC}}(P) \mathrel{P_j} \mu^*_j$. Then $\mu^*$ does not Pareto dominate $\varphi^{\mathrm{TTC}}(P)$ at $P$, so the definition of $\varphi^{\neg \MAR}$ implies that $\varphi^{\neg \MAR}_i(P) = \varphi^{\mathrm{TTC}}(P)$ and $\varphi^{\neg \MAR}_i(P') = \varphi^{\mathrm{TTC}}(P')$. Thus \TP for $\varphi^{\mathrm{TTC}}$ implies that $P'_i$ is not a profitable manipulation.
        \item Suppose $\mu_j^* \mathrel{R_j} \varphi_j^{\mathrm{TTC}}(P)$. Then the definition of $\varphi^{\neg \MAR}$ implies that $\varphi^{\neg \MAR}_i(P)$ is the most-preferred bundle in $\{\mu_i^*, \varphi_i^{\mathrm{TTC}}(P)\}$ according to $P_i$. If $P' \in \mathcal{R}^*$, then $\varphi_i(P') \in \{\mu_i^*, \varphi_i^{\mathrm{TTC}}(P)\}$, which implies $\varphi^{\neg \MAR}_i(P) \mathrel{R_i} \varphi^{\neg \MAR}_i(P')$. If $P' \notin \mathcal{R}^*$, then \TP for $\varphi^{\mathrm{TTC}}$ implies that
        $$
        \varphi^{\neg \MAR}_i(P) \mathrel{R_i} \varphi_i^{\mathrm{TTC}}(P) \mathrel{R_i} \varphi^{\mathrm{TTC}}_i(P') = \varphi^{\neg \MAR}_i(P').
        $$
    \end{casenv}
\item Suppose $P \in \mathcal{R} \setminus \mathcal{R}^*$. If $P' \in \mathcal{R} \setminus \mathcal{R}^*$, then \TP for $\varphi^{\mathrm{TTC}}$ implies that
$$
\varphi^{\neg \MAR}_i(P) = \varphi_i^{\mathrm{TTC}}(P) \mathrel{R_i} \varphi^{\mathrm{TTC}}_i(P') = \varphi^{\neg \MAR}_i(P').
$$
Furthermore, it is not possible that $P' \in \mathcal{R}^*$. Indeed, $P' \in \mathcal{R}^*$ would imply $P'_i|_O \neq P_i|_O$, in which case $P'_i|_O$ is obtained from $P_i|_O$ by dropping a \emph{nonempty} tail subset $X \subseteq O \setminus \omega_i$. In particular, this means $\min_{P'_i}(O) \notin \omega_i$, a contradiction.
\hfill $\diamond$
\end{casenv}
\end{example}

\newpage

\bibliographystyle{econ-econometrica}
\phantomsection\addcontentsline{toc}{section}{\refname}\bibliography{references}


\end{document}